\documentclass[letterpaper,twocolumn,10pt]{article}
\usepackage{fullpage}

\input{macros_IEEE}

\usepackage{tikz}
\usepackage{wrapfig}
\usepackage{amsmath}

\begin{document}

\date{}

\def\name{\textsf{Piquant$\varepsilon$}}
\title{\name: Private Quantile Estimation in the Two-Server Model}

\author{
{\rm  Hannah Keller}\\
Aarhus University
\and
{\rm Jacob Imola}\\
BARC, University of Copenhagen
\and 
{\rm Fabrizio Boninsegna}\\
University of Padova
\and
{\rm Rasmus Pagh}\\
BARC, University of Copenhagen
\and
{\rm Amrita Roy Chowdhury}\\
University of Michigan
}

\maketitle

\begin{abstract}
Quantiles are key in distributed analytics, but computing them over sensitive data risks privacy. Local differential privacy (LDP) offers strong protection but lower accuracy than central DP, which assumes a trusted aggregator. Secure multi-party computation (MPC) can bridge this gap, but generic MPC solutions face scalability challenges due to large domains, complex secure operations, and multi-round interactions.

We present \name, a system for privacy-preserving estimation of multiple quantiles in a distributed setting without relying on a trusted server. \name~operates under the malicious threat model and achieves accuracy of the central DP model. Built on the two-server model, \name~uses a novel strategy of releasing carefully chosen intermediate statistics, reducing MPC complexity while preserving end-to-end DP. Empirically, \name~estimates 5 quantiles on 1 million records in under a minute with domain size $10^9$, achieving up to $10^4$-fold higher accuracy than LDP, and up to $\sim 10\times$ faster runtime compared to baselines\footnote{Code available at https://github.com/hjkeller16/Piquante}.
\end{abstract}

\section{Introduction}
Quantiles are a fundamental tool for understanding data distributions and are widely used in distributed analytics, where data is spread across many clients. They enable applications such as system performance monitoring, telemetry, and financial risk assessment~\cite{quantile,dean2013tail}. 
However, releasing quantiles over sensitive data can inadvertently leak private information—for instance, reporting the median income may expose individual salaries, revealing median medical expenses could hint at specific health conditions, or showing the top spenders in a dataset could identify individual customer behavior.
Local Differential Privacy (LDP) offers a compelling solution by providing strong privacy guarantees \textit{without} a trusted  server. This makes LDP particularly well-suited for distributed analysis, where centralized trust is often impractical or undesirable, and is already deployed at scale by companies like Google~\cite{Rappor1}, Apple~\cite{Apple}, and Microsoft~\cite{Microsoft}.

Despite these advantages, LDP comes with a well-documented trade-off: a significant gap in accuracy compared to central DP~\cite{error1,error2,error3}. In the central model, a \textit{trusted} server collects raw data from the clients and then applies a differentially private mechanism on the centralized dataset, achieving better utility. In contrast, the local model requires each client to \textit{individually} randomize their data  before transmission. In the context of quantiles, prior work shows an error bound of $O(\frac{\sqrt{n}\log(|\B|)^2}{\epsilon })$\footnote{\label{f1} For simplicity we omit a term that depends on the failure probability~$\beta$.}~\cite{edmonds2020power, hierarchical_LDP} in the local model where $\B$ is the data domain, while the central DP estimates can achieve error $O(\frac{\log(|\B|)}{\epsilon} + \frac{\log^2 m}{\epsilon})$\footnoteref{f1}~\cite{imola2025differentiallyprivatequantilessmaller}.
A natural strategy for bridging this gap is to leverage secure multi-party computation (MPC) to design a distributed DP protocol. In particular, the exponential mechanism~\cite{EM} is known to be the optimal approach for estimating a single quantile under central DP. Thus, implementing the exponential mechanism securely via MPC enables accurate and private quantile estimation in a distributed setting—without requiring a trusted aggregator. 

While appealing in theory, deploying this approach using generic MPC tools in practice suffers from several challenges. 
First, a generic MPC solution would require executing a multi-round protocol across all participating client devices. At the scale of millions of users, such protocols become entirely impractical due to network latency, synchronization overhead, and frequent client dropouts.
Second, the exponential mechanism selects an output from the domain based on a utility score that reflects the quality of the result. For quantile estimation, the output domain is effectively the entire domain of data—which can be extremely large. 
Another issue is that quantiles are rank-based statistics, and computing utility scores for the exponential mechanism requires sorting the dataset. Secure sorting under MPC is notoriously expensive and becomes prohibitively costly at scale.
Finally, estimation of \textit{multiple} quantiles further compounds the complexity. Naively repeating the estimation \( m \) times not only increases computational cost but also degrades the utility due to the composition of multiple DP mechanisms. 

A line of prior work~\cite{CryptE,10.1145/2818000.2818027,cryptoeprint:2014/482} has focused on developing general-purpose frameworks that bridge the gap between the central and local models of DP. However, the generality of these solutions comes at a cost: they inherit the computational and scalability challenges discussed above. 
The closest line of work that specifically targets quantile estimation—particularly median computation—includes~\cite{BohlerK20, Kerschbaum}. However, both approaches suffer from critical limitations. The protocol in~\cite{BohlerK20} is designed for computing a single quantile over the joint dataset of exactly two parties, each of whom holds an entire dataset. This setting does not generalize to distributed settings where each client contributes only a single data point. Moreover, the protocol assumes a semi-honest threat model and does not scale well with large domain sizes. The approach in~\cite{Kerschbaum} improves on domain-size dependency but still assumes the semi-honest model, supports only a single quantile estimate, and requires an interactive protocol with 
$O(\log |\B|)$ communication rounds per client—making it impractical for large-scale deployment.
\\\noindent\textbf{Our Contributions.}
In this paper, we present \name, a system for the privacy-preserving estimation of multiple quantiles in a distributed setting without relying on a trusted central server. We focus on the practically relevant case of estimating up to  $o(\sqrt{n})$ quantiles (e.g., about 20 quantiles for a dataset with one million records). Such quantile resolutions are widely used in many real-world tasks: income statistics often use $m=4$~\cite{income}, clinical reference ranges such as birth-weight statistics typically use $m=9$~\cite{CDC}, and educational assessments may use $m=20$ to define grading thresholds~\cite{GRE}.
\name~operates under the \emph{malicious} threat model and achieves accuracy guarantees matching those of central DP. \name~builds on the widely adopted two-server model, where each client secret-shares their data between two non-colluding servers. This model has been extensively studied in distributed learning~\cite{HPI, CryptE, Bell22} and successfully deployed in real-world applications~\cite{deploy1,deploy2}. It is also currently under standardization by the IETF~\cite{ietf}. 

At a high level, rather than consuming the entire privacy budget to directly estimate the requested quantiles, \name~strategically allocates part of the budget to release intermediate statistics. The core novelty of \name~lies in the careful design and selection of these intermediate statistics, which substantially reduces the complexity of the required MPC operations without sacrificing the accuracy of the final quantile estimates. As a result, \name~eliminates all three major bottlenecks in secure quantile estimation—dependency on input domain size, dataset size, and the number of target quantiles. \name's~protocol is non-interactive from the client’s perspective, requiring only a single one-shot communication of secret-shares to the servers—yielding a system that is efficient and scalable. Empirically, \name~estimates 5 quantiles on 1 million records in under a minute with domain size $|\B|=10^9$, achieving  $10^4$-fold better accuracy than LDP with up to $\sim 10\times$ faster runtime 
compared to baselines.

\section{Background}\label{sec:background}

\subsection{Cryptographic Primitives}\label{sec:crypto}
\noindent \textbf{Linear Secret Shares (LSS).} 
Linear secret shares (LSS) is an MPC technique that allows mutually-distrusting parties to compute over secret inputs. We use $[x]$ to denote a linear secret sharing of an input $x \in \mathbb{F}$. Each party $P_i$ holds a share $[x]_i$ such that $\sum_{i=1}^k [x]_i=x$. LSS supports local linear operations. 
\\\noindent \textbf{Authenticated Secret Shares (ASS).} 
We adopt the SPDZ-style~\cite{Spdz2,Spdz1} authenticated secret shares (ASS), which ensures the integrity of the shared secrets. ASS augments the shares with an additional sharing of an information-theoretic message authentication code (IT-MAC). Specifically, each party shares $[\gamma]$ for a secret MAC key $\gamma \in \mathbb{F}$, and for a shared value $[x]$, they also share the corresponding MAC $[\gamma \cdot x]$. We denote $\langle x \rangle = ([x], [\gamma \cdot x])$
to be the ASS for a secret $x$ and $\langle x \rangle_i = ([x], [\gamma \cdot x]) \in \mathbb{F}^2$ as the ASS held by party $P_i$. ASS also supports local linear operations. $\textsf{Share}(x)$ denotes the protocol of generating and distributing ASS of a value $x$ (App.~\ref{app:ass}).

\begin{table}
\centering
\scalebox{0.9}{\begin{tabular}{lp{6cm}}
\toprule\\
\textbf{MPC Protocol} & \textbf{Output/ Functionality}\\
\midrule
$\Rec(\langle x\rangle)$ & x, reconstructed from $\langle x\rangle$\\
$\Mult(\langle x \rangle, \langle y \rangle)$ & $\langle x \cdot y\rangle$\\
$\Abs(\langle x\rangle)$ & $\langle |x| \rangle$\\
$\Trunc(\langle x \rangle,y)$ & $\lfloor x/2^y
\rfloor$, where $y$ is public\\
$\textsf{Equal}(\langle x \rangle, \langle y \rangle)$& $\langle 1\rangle$, if $x==y$, $\langle 0 \rangle$ otherwise\\
$\Cmp (\langle x \rangle, \langle y \rangle)$ & $\langle 1\rangle$, if $x<y$, $\langle 0 \rangle$ otherwise
\\ $\textsf{Rand}(y)$ & $\langle r \rangle$ with uniform random $y$-bit value $r$
\\
\bottomrule
\end{tabular}}
\caption{Basic MPC Protocols~\cite{cryptoeprint:2012/405} used in \name. We prefix protocols for floats with $\textsf{FL}$.}
\end{table}
\noindent\textbf{Secure Computation.} 
Let $\Pi$ be a $k$-party secure computation protocol, with up to $k' < k$ parties corrupted by a static adversary $\mathcal{A}$. Let $X$ denote the honest parties’ inputs and $\lambda$ the computational security parameter. We denote by $\texttt{REAL}_{\Pi,\mathcal{A}}(1^\lambda, X)$ the random variable representing the view of $\mathcal{A}$ and the output of the honest parties during an execution of $\Pi$, when they hold $X$. To formally characterize the security guarantees of $\Pi$, we define an ideal world in which a trusted third party evaluates an ideal functionality $\mathcal{F}$ for $n$ clients in the presence of $\mathcal{A}$. $\mathcal{F}$ may be interactive: it can send intermediate values to $\mathcal{A}$ and receive additional inputs (depending on the specifics of $\mathcal{F}$). In particular, $\mathcal{A}$ may send an \texttt{ABORT} message to $\mathcal{F}$, which prevents the honest parties from receiving any output. After the interaction, $\mathcal{A}$ may compute and output an arbitrary function of its view. We denote by $\texttt{IDEAL}_{\mathcal{F}, \mathcal{A}}(1^\lambda, X)$ the random variable corresponding to the adversary’s output and the output of the honest parties in this ideal execution.


\begin{definition} (Secure computation). $\Pi$ securely computes $\F$
if for every polynomial-time real-world adversary $\calA$ there is a
polynomial-time ideal-world adversary $\calA'$
such that for all X,
the distributions $\texttt{REAL}_{\Pi,\calA}(1^\lambda
, X)$ and $\texttt{IDEAL}_{\F,\calA'} (1^\lambda
, X)$
are computationally indistinguishable.
\end{definition}

\subsection{ Differential Privacy}\label{sec:DP}
Here, we provide background on differential privacy. Two datasets of size $n$ are neighboring if one can be obtained from the other by substituting a single user.
\begin{definition}[Differential Privacy]
    A randomized mechanism $\mathcal{M}: \mathcal{X}^n \rightarrow \mathcal{Y}$ satisfies $(\epsilon, \delta)$-DP if, for any neighboring datasets $X\sim X'$ and for any $Y\subseteq \mathcal{Y}$ we have that 
    \begin{equation}\label{eq:dp}
    \mathrm{Pr}[\calM(X)\in Y] \leq e^{\epsilon} \mathrm{Pr}[\calM(X')\in Y] + \delta \enspace .
    \end{equation} 
    We abbreviate the property (\ref{eq:dp}) as $ \calM(X) \approx_{\epsilon,\delta} \calM(X')$.
\end{definition}
\begin{definition} An ideal functionality $\F$ satisfies $(\epsilon, \delta)$-DP if for all $\calA$, $\lambda$, and $\X \sim X'$,
\begin{gather} 
\texttt{IDEAL}_{\F,\calA} (1^\lambda
, X) \approx_{\epsilon,\delta}\texttt{IDEAL}_{\F,\calA} (1^\lambda
, X') \enspace .
\end{gather} 
\end{definition}

Next, we  define computational DP
for cryptographic protocols by adapting the notion of $\text{SIM}^+$-CDP \cite{computationalDP}. 
\begin{definition}[Computational Differential Privacy]
    A protocol $\Pi$ satisfies $(\epsilon, \delta)$-computational
differential privacy if there is an $(\epsilon, \delta)$-DP ideal
functionality $\F$ such that $\Pi$ securely computes $\F$.\label{def:compDP} 
\end{definition}



\noindent We now present two mechanisms that will be used in \name.
\\\noindent\textbf{Exponential Mechanism.} The exponential mechanism is a classic differentially private algorithm for selecting an output from a set of candidates based on a utility function $\u: \mathcal{X}^n \times \calY \to \mathbb{R}$. $\u$ quantifies the candidate by assigning a real-valued score. It is designed to prefer ``good" outputs—those with higher utility—by assigning them higher probabilities of being selected. 
\begin{definition} [Exponential Mechanism~\cite{mcsherry2007mechanism}] For any utility function $\u : (\mathcal{X}^n \times \calY )\mapsto \mathbb{R}$, a privacy parameter $\epsilon$ and a dataset $X \in \mathcal{X}^n$, the exponential mechanism $\EM_\u(X, \epsilon)$ releases an element $y\in \mathcal{Y}$ according to 
\begin{equation*} 
\mathrm{Pr}\big[\EM_\u(X, \epsilon) = y] = \exp \left(\frac{\epsilon u(X,y)}{2	\bigtriangleup \u}\right) \Big/ \underset{y' \in \calY}{\sum} \exp \left(\frac{\epsilon \u(X,y)}{2	\bigtriangleup \u}\right), 
\end{equation*}
where $\bigtriangleup u$ is the sensitivity of the utility function defined as $\bigtriangleup u = \max_{X\sim X'}\max_{y\in \mathcal{Y}}|u(X, y)-u(X', y)|$,
 is $(\varepsilon, 0)$-DP.
\end{definition}

\noindent\textbf{Continual Counting.} 
A continual counting mechanism is a differentially private method for releasing prefix sums $\sum_{i=1}^{t} x_i$ for any $t < T$, where $x_i \in {0,1}$. A naive approach adds independent Laplace noise to each of the $T$ prefix sums, incurring maximum error $\tilde{O}(T / \varepsilon)$. The \emph{binary mechanism} \cite{continual-counting, private-continual-statistics} improves this to $O(\log^2 T / \varepsilon)$ by adding correlated Laplace noise derived from a binary tree over prefix sums. 
We denote the resulting noise distribution by $\CC_{\varepsilon}(t) \in \mathbb{Z}^t$, which ensures privacy w.r.t. neighboring inputs $\mathbf{x}, \mathbf{x}'$ (with $|\mathbf{x} - \mathbf{x}'| \leq 1$). We defer the full details to App.~\ref{app:cont-count}.

\begin{lemma}\label{lemma:CC}(From~\cite{chan2011private})
    There exists a noise distribution $\CC_\varepsilon(t)$ (written $\CC$ for short) supported on $\mathbb{Z}^t$ such that 
    \squishl
        \item For all neighboring partial sum vectors $\mathbf{s}, \mathbf{s'} \in \Z^m$ and vectors $\mathbf{v} \in \Z^t$, we have 
    \[
        \Pr_{\boldsymbol{\eta} \sim \CC}[\mathbf{s}+\boldsymbol{\eta}=\mathbf{v}] \leq e^{\epsilon} \Pr_{\boldsymbol{\eta} \sim \CC} [ \mathbf{s}' + \boldsymbol{\eta} = \mathbf{v}]
   \]
    \item For all $\beta > 0$, $\Pr_{\boldsymbol{\eta} \sim \CC}[\|\boldsymbol{\eta}\|_\infty \geq \frac{3}{\epsilon}\log(t) \log(\frac{2t}{\beta})] \leq \beta$.
    \squishe
\end{lemma}

\section{Problem Statement}\label{sec:problem}

\textbf{Notations.} Given dataset $X = \{x_1, \ldots, x_n\}$ in discrete data domain $\B$, we denote by $x_{(i)}$ the $i$th smallest element in $X$. Vectors are denoted in \textbf{bold}.
For $q \in [0,1]$, the $q$th quantile of $X$ is defined as $x_{(r)}$, where $r = \lfloor qn \rfloor$. The \textit{quantile error} measures the difference between the desired quantile rank $q$ and actual rank of the quantile estimate $z$. For $z \in \mathcal{B}$, the error at point $z$ and quantile $q$ is: 
\begin{equation} 
    \mathsf{err}_{X,q}(z) =\left\lvert q n - \sum_{x \in X} \textbf{1}[x \leq z]\right\rvert \label{eq:utility:quantile}
\end{equation}
\noindent\textbf{Setting.} We consider datasets $X$ where each point is held by one of $n$ clients. The goal is to compute a given set of quantiles $Q=\{q_1, \cdots, q_m\} \in (0,1)^m$. We assume all data points in $X$ are unique, implying that $|x_i - x_j| \geq 1$ for all $i \neq j$. This is a mild assumption that makes our theoretical results easier to state, and can be easily enforced in practice. \footnote{The gap can be enforced by expanding the domain to $|\B|n$ and appending $\log(n)$ unique lower-order bits per client. If unmet, only utility—not privacy or security—is affected.}

In local DP, quantile estimation incurs $\mathcal{O}\left(\sqrt{n}\log^2 |\B| /\epsilon \right)$ error~\cite{edmonds2020power, hierarchical_LDP}, which is significantly worse than the $\mathcal{O}\left(\log |\B| /\epsilon  + \log(m)^2/\epsilon \right)$ achievable under central DP~\cite{imola2025differentiallyprivatequantilessmaller}. An intermediate model, known as shuffle DP, assumes a trusted shuffler that randomly permutes client messages before aggregation, yielding privacy amplification~\cite{cheu_distributed_2019,bittau_prochlo_2017}. In this setting, the best bound on error is $\mathcal{O}\left(\log^2 |\B|/\epsilon\right)$~\cite{shuffle-hiding}, which can still grow too quickly for larger data domain sizes. Our goal is to match the accuracy of the central model while retaining the local model’s assumption of no trusted party, by leveraging an MPC protocol.

 We consider a setting with two \textit{non-colluding, untrusted} servers, $\Ser_0$ and $\Ser_1$. One server may be operated by an entity interested in computing the quantiles, while the other may be an independent third party. Existing services like Divvi Up~\cite{divvi-up}\footnote{Operated by the non-profit Internet Security Research Group.} provide infrastructure to support such non-colluding  servers. Each client splits their private input $x_i$ into ASS and sends them to the two servers. This is a one-time operation, after which clients can go offline. The servers then engage in an MPC protocol to compute the desired quantiles.  
\\\noindent\textbf{Threat Model.}
We consider a probabilistic polynomial-time (PPT) adversary $\mathcal{A}$ that actively corrupts at most one server and an arbitrary number of clients. These server and clients are assumed to be \textit{malicious}—that is, they may deviate arbitrarily from the protocol specification.
Our privacy guarantees ensure that the servers learn only information that is differentially private with respect to the clients' inputs.

\section{Technical Overview}

\subsection{Naive Approach}
Under central DP, a very effective mechanism for computing a single quantile is the exponential mechanism (Sec.\ref{sec:DP}), where the utility function is given by the quantile error~Eq.\eqref{eq:utility:quantile}. A naive strategy to privately compute multiple quantiles in our setting is to instantiate this mechanism under secure computation and invoke it independently for each of the 
$m$ quantiles. However, this approach introduces several challenges:

\squishl
\item First, the exponential mechanism requires evaluating selection probabilities over the entire output domain $\B$, resulting in time complexity linear in $|\B|$. Since $\B$ is typically much larger than the number of unique input values, and in some cases even larger than the total number of users, this introduces a significant computational bottleneck.
\item Second, in the context of estimating multiple quantiles, the naive approach has two major drawbacks. First, it incurs a computational cost of at least $O(nm)$, due to evaluating each of the $m$ quantiles independently over $n$ inputs. Second, from a utility perspective, applying DP via sequential composition increases the overall error by a factor $m$.\footnote{Although prior work improves upon the naive approach by computing quantiles recursively over a binary tree—resulting in a more favorable dependence on the privacy budget due to composition—implementing this method under MPC still incurs a computational overhead of $\mathcal{O}(n(m+\log n))$.  We use this tree-based protocol as a baseline in Sec.~\ref{sec:eval}.}

\item Third, computing each candidate’s utility score requires its rank in $X$, which entails sorting the \textit{entire} dataset. Secure sorting under secret sharing is costly since comparison is a non-linear operation, with each comparison incurring heavy computation and multiple communication rounds~\cite{edaBit}. 

\item Finally, computing the selection probabilities involves performing $|\B|$ exponentiations over floating-point numbers. This is costly in secure computation due to its non-linearity, further exacerbating the overall computational burden.
\squishe

\subsection{\name: Key Ideas}\label{sec:key ideas}
In light of the above challenges, \name~does not release the quantiles directly using the entire privacy budget. Instead, our approach allocates a portion of the budget to release some additional intermediate statistics from the joint dataset $X$ (Fig.~\ref{fig:name}). The novelty of \name~lies in the strategic selection of this additional, differentially private leakage, which significantly reduces the complexity of the required MPC operations -- without compromising the accuracy of the final quantile estimates. In what follows, we elaborate on the specifics of this strategy and how it addresses the aforementioned challenges.
\\\noindent \name~is based on four key ideas:
\\\noindent\textbf{1. Remove computational dependence on the domain size $|\B|$.}
Our first idea builds on the properties of the utility score for quantile estimation. Let $\I$ denote an interval in $[\B]$ such that no private data point from the dataset $\X$ is contained in $\I$, i.e., $\forall x \in \X, x \notin \I$. Now, observe that for every point $t \in \I$, the utility score $\u_q(\X, t)$ for the $q$th quantile is exactly the same, since ${\sum \mathbf{1}_{x\in \X}[x \leq t] } =  \sum \mathbf{1}_{x\in \X}[x \leq t'], \forall (t, t') \in \I \times \I$. As a consequence, instead of having a separate selection probability for each element in $\I$, we can compute a weighted probability for each interval, where the weights are proportional to the size of the interval.
The exponential mechanism can thus be rewritten to the following two-step sampling strategy using the techniques from \cite{Smith11}:
\squishl \item Sort all the data points in $\X \in \B^n$. 
\item Let $\I_k = [x_{(k-1)}, x_{(k)}]$, where $k = 1, \cdots, n + 1$, and $x_{(0)} = 0$, $x_{(n+1)} = \d$, be the set of $n + 1$ intervals between the sorted data points in $\X_{\text{sort}}$. We sample an interval from this set of intervals, where the probability of sampling $\I_k$ is proportional to \begin{gather*} \mathrm{Pr}[\A(\X) = \I_k] \propto \exp\Big(\frac{\epsilon \u_q(\X, \I_k)}{2}\Big) \cdot (x_{(k)} - x_{(k-1)}), \end{gather*} 
Here $\u_q(\X, \I_k)$ refers to the utility of the interval (or equivalently any point within the interval).
\item Return a uniform random point from the sampled interval. \squishe
As a result, \name's computation and communication costs are independent of the domain size $|\B|$. 

\begin{mybox2}{Ideal Functionality: $\FEM$}
\textbf{Input:} Sorted private dataset $X \in \B^n $ with $\B=[\d]$; Queried quantile $q \in (0,1)$; Privacy parameters $\epsilon$; 
\\
\textbf{Output:} Noisy quantile estimate $v$
\begin{enumerate}[nosep]
\item Let the sorted dataset be $\Xs = \{x_{(1)},\cdots, x_{(n)}\}$
\item $\I_1=[1,x_{(1)}]$,  $\I_{n+1}=[x_{(n)},d]$
\item \textbf{For} $i\in [n]$
\item \hspace{0.5cm} $\mathcal{P}(\I_i)=\frac{\exp (\frac{\epsilon u_q(X, \I_i)}{2}) \cdot (x_{(i)}-x_{(i-1)})}{\sum_{j=1}^{n+1} \exp (\frac{\epsilon u_q(X, \I_j)}{2}) \cdot (x_{(j)}-x_{(j-1)})}$
\item \textbf{End For}
\item Sample $\I^*\sim \mathcal{P}$ and return an element $v$ chosen uniformly at random from $\I^*$
\end{enumerate}
\end{mybox2}

 \noindent\begin{minipage}{\columnwidth}
\captionof{figure}{Ideal Functionality: $\F_{\EM}$}\label{fig:IF:EM}
\end{minipage}

\noindent\textbf{2. Improve error dependency on the number of quantiles.} Our second key idea arises from the observation that, rather than using the entire dataset for estimating each of $m$ quantiles, it is sufficient to operate on specific \textit{slices} of the dataset in sorted order. In particular, to privately estimate the quantile $q$, we only need access to a small neighborhood around its corresponding rank $\rq$—$\{x_{(\rq - h)}, \ldots, x_{(\rq + h)}\}$ for some sufficiently large integer $h$. This insight allows us to partition the sorted dataset into $m$ \textit{disjoint} slices—one for each quantile. We then estimate all $m$ quantiles in \textit{parallel} by invoking a (secure) exponential mechanism independently on each slice (Fig. \ref{fig:sEM}). By applying the parallel composition theorem (Thm.~\ref{thm:parrallel}), this strategy improves utility by a factor of $m$ over the naive approach. However, as noted in recent work~\cite{imola2025differentiallyprivatequantilessmaller}, a technical challenge arises: deterministically computing these slices can lead to instability in the partitioning, violating the assumptions required for parallel composition (see Sec.~\ref{sec:IF:slicing} for details). To mitigate this,~\cite{imola2025differentiallyprivatequantilessmaller} proposes to perturb the quantiles themselves, thereby making them private. Unfortunately, this is not suitable for our setting. Making the quantiles private would require computing selection probabilities using secure floating-point exponentiations--which, as discussed earlier,  we seek to avoid. 
\begin{figure}
\centering
\includegraphics[width=0.9\linewidth]{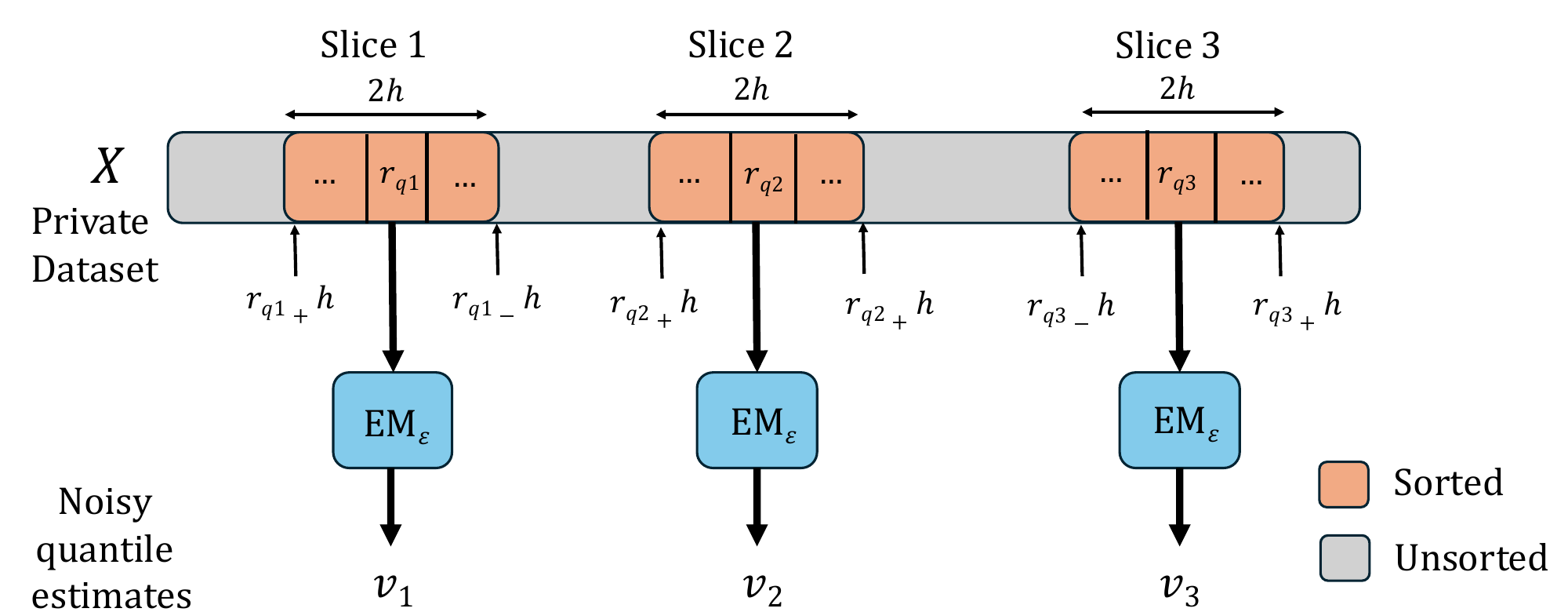} 
\caption{Overview of \SEM~for three quantiles $Q=\{q_1,q_2,q_3\}$. $r_{q_1}$ denotes the element with rank in $\lfloor q_i n\rfloor$ in the dataset $\X \in \B^n$.}

\label{fig:sEM}
\end{figure}
\begin{figure*}
\centering
\includegraphics[width=0.9\linewidth]{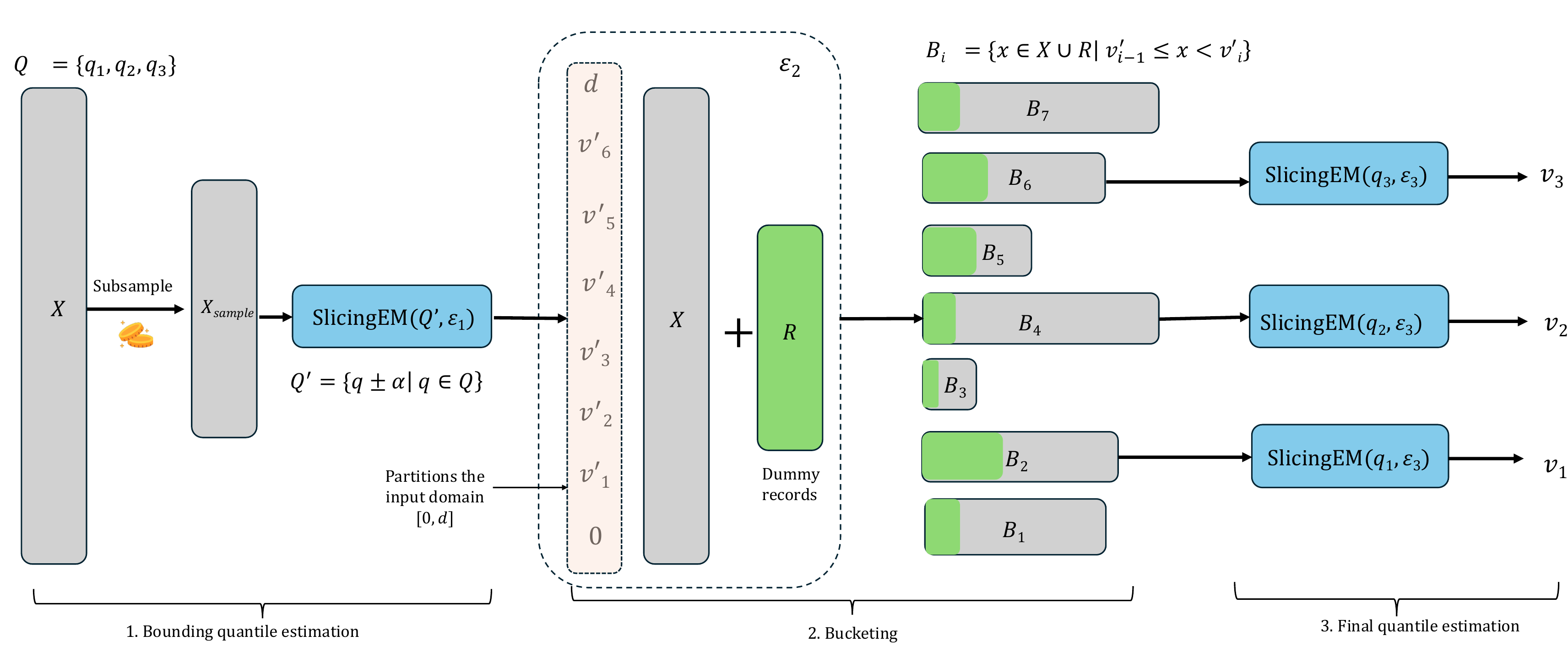}
\caption{Overview of the ideal functionality of \name. \name~securely implements this in the two-server model.}

\label{fig:name}
\end{figure*}
\name~takes a different approach. Instead of perturbing the quantiles, we introduce a novel method that performs a \textit{random shift} of the slicing boundaries to ensure compatibility with parallel composition. This strategy avoids making the quantiles private and eliminates the need for costly secure floating-point operations. Furthermore, we develop an efficient technique to implement this shifting under MPC—making our approach a win-win in terms of both utility and efficiency. We refer to this method as \SEM.
 \\\noindent\textbf{3. Improve scalability with dataset size $n$.} \name~tackles scalability with $n$ with two complementary techniques. First, we observe that for the slicing-based approach it suffices to perform a \textit{partial sort}—only the elements within each of the $m$ slices need to be sorted, reducing the number of secure comparisons from $O(n\log n)$ to $O(n\log m)$ (Thm. \ref{coro:partial-sort}; see Fig. \ref{fig:sEM}). 
 Second, \name~further improves scalability by \textit{reducing the input size for each partial sort}.   The high-level idea is to isolate the individual slices via a two-phase quantile estimation strategy. In the first phase, for every target quantile $q \in Q$, \name~allocates a small portion of the privacy budget to estimate a pair of \textit{bounding quantiles} $q \pm \alpha$. These bounds are used to filter the input dataset into a bucket: \[ \Bu_q = \{x \in \X \mid \a_{q-\alpha} \leq x < \a_{q+\alpha} \}.\] 
In other words, this bucket contains all (unsorted) data points corresponding to the slice relevant for estimating $q$. Finally, in the second phase, \name~applies  \SEM~to each $\Bu_q$ to estimate the target quantile $q$—now using a much smaller subset of the input. Translating this idea into a secure implementation under MPC introduces two key challenges:
\squishl
    \item First, this approach might appear to introduce a circular dependency since every queried quantile now requires  additional estimates for $q\pm \alpha$. However, the key insight is that these estimates need not be highly accurate—they are only used to filter out irrelevant data. To reduce overhead, \name~computes these bounds using \SEM~on a small \textit{subsample} $\Xsam$ of the dataset, yielding sufficiently accurate yet much cheaper estimates. 
    \item Second, securely populating buckets under MPC is non-trivial. Naively padding each bucket with dummy records to size $n$ to hide its true cardinality defeats the efficiency gains of bucketing. \name~mitigates this by leaking a differentially private estimate of the bucket sizes. This allows \name~to add only the minimal number of dummy records needed to satisfy DP—far fewer than with full padding. To further optimize utility, we introduce a new noise distribution for sampling the dummy records (see Sec.\ref{sec:IF:bucketing}).  We refer to this step as \Buck.
\squishe

 Thus, by leaking a some additional intermediate statistics—(1) low-accuracy estimates for bounding quantiles  $Q' = \{q \pm \alpha \mid q \in Q\}$, (2) noisy bucket sizes—\name~substantially reduces computational overhead while preserving an end-to-end DP guarantee. In particular, \name~ensures that in both phases, \SEM~is applied to significantly smaller inputs (the subsample $\Xsam$ and the filtered buckets $\Bu_q$), thereby improving scalability with $n$.
\\\noindent\textbf{4. Avoid Secure Floating-Point Exponentiation.} The final challenge is securely computing the selection probabilities for the intervals in the buckets. Specifically, we need to compute $\exp(\epsilon\u_q(X,\I)/ 2)$, which entails knowing the rank of the interval $\I$ (or equivalently, any point $t \in \I$) w.r.t to the \textit{entire} dataset $\X$.  At first glance, this seems to necessitate computing the rank securely over the entire dataset and then performing a secure exponentiation. \name~circumvents this as follows.  As part of the bucketing strategy described earlier, \name~releases a differentially private estimate of the size of each bucket.  This is, in fact, equivalent to revealing the differentially private rank of the edge points of each bucket, which allows us to compute the utility $\u_q(\X,\I)$ and, subsequently,  exponentiate it \textit{in the clear}. As a result, \name~completely avoids expensive secure floating-point exponentiations.

\section{\name: Ideal Functionalities}
In this section, we describe the ideal functionalities of the key subroutines that make up \name. 
We adopt a \textit{bottom-up} approach: Sec. \ref{sec:IF:slicing} begins with the ideal functionality for \SEM, followed by the functionality for $\Buck$ in Sec. \ref{sec:IF:bucketing}. Finally, in Sec.~\ref{sec:everything}, we integrate these components and present the full ideal functionality of \name.

These functionalities explicitly model interactions with an adversary $\calA$, denoted by $\ad$ in the pseudocode. Both the functionality and the adversary contribute noise: the functionality samples noise from the correct distribution, while $\calA$—depending on whether it is passive or active—either follows the same distribution or injects arbitrary noise. This mirrors our two-server MPC protocols, where both servers are required to contribute noise, but one of the servers may be corrupted and controlled by $\calA$. Additionally, $\calA$ is allowed to abort the functionality at any point.  Consequently, our privacy  and utility analyses extend to the MPC protocols instantiated by \name, which securely implements these functionalities in the two-server setting.

\subsection{\SEM}\label{sec:IF:slicing}
Inspired by a recent work~\cite{imola2025differentiallyprivatequantilessmaller}, \SEM~is a mechanism for estimating multiple quantiles under (central) DP. 
Given a dataset $X$ and a set of target quantiles $Q = {q_1, \ldots, q_m} \in (0,1)^m$, the key idea is to construct $m$ slices of the sorted dataset—each centered around the rank corresponding to a quantile—and apply the exponential mechanism independently to each slice. 
Formally, each slice $S_i$ is given by: 
\begin{equation*}
  S_i =\Xs[c_i, d_i],\quad c_i=r_{q_i}-h,\quad d_i = r_{q_i} +h
\end{equation*}
where $r_{q_i}=\lfloor q_i n\rfloor$
is the rank corresponding to quantile $q_i$ and $h\in \mathbb{N}$
is a parameter chosen based on the utility guarantee of the exponential mechanism. This parameter is chosen in such a way that the resulting slices are mutually disjoint.

However, as noted by prior
work~\cite{imola2025differentiallyprivatequantilessmaller,10.1145/3564246.3585148, 10.1109/FOCS.2015.45} a technical challenge arises: to leverage parallel composition,  the slices must remain nearly unchanged—or \textit{stable}—under adjacent datasets $X$ and $X'$. Unfortunately, the deterministic strategy mentioned above for slicing is not stable. To illustrate, consider a neighboring dataset $X'$ obtained by adding a data point in $X$. For some $t \in [m]$, this addition triggers a cascading change in the slices $S'_t, \ldots, S'_m$, causing a naive analysis using standard composition to inflate the error by a factor of $O(m)$.

To mitigate this,~\cite{imola2025differentiallyprivatequantilessmaller} observes that these changes are structured. In particular, for all $i \geq t$, each slice pair \( \{S_i, S'_i\} \) is shifted by exactly one position, i.e., $
 S_i = \Xs'[c_i-1, \d_i-1].$ This structure enables the use of the continual counting noise mechanism, as property (1) of Lemma~\ref{lemma:CC} allows these shifts to be masked.\footnote{The structure of the shifts mirrors the behavior of a partial sum in $\CC$.}  Concretely, \cite{imola2025differentiallyprivatequantilessmaller}~perturbs each target rank $r_{q_i}$ (equivalently, the queried quantile) with noise $\eta_i$ as $\tilde{r}_{q_i} = r_{q_i} + \eta_i\quad \text{where} \quad \boldsymbol{\eta}\gets\CC_\epsilon(m)$
and generates the slices based on them. 
By Lemma~\ref{lemma:CC}, the resulting noise ensures that the slicing remains stable under the addition of a single data point. Intuitively, this allows the final outputs $\EM(\tilde{S}_1), \ldots, \EM(\tilde{S}_m)$\footnote{For brevity, we omit the additional inputs to $\EM$ in this notation.} to be analyzed under parallel composition — incurring the privacy cost of applying $ \EM$ only once. 
A technicality introduced by the continual counting noise is that the input quantiles must be separated by roughly $\frac{2}{n}(w+h+1)$, which is asymptotically larger than the final error guarantee. For the setting we are interested in where $m$ is much smaller than $n$, and especially for quantiles which are equally spaced within $[0,1]$, this assumption is mild and is usually met. 

\begin{mybox2}{Ideal Functionality:  $\FSEM$}
\textbf{Input.}
Private dataset \scalebox{0.9}{$X \in \B^n$}; Queried quantiles \scalebox{0.9}{$Q = \{q_1, \cdots, q_m\} \in (0,1)^m$};  Privacy parameters \scalebox{0.9}{$\epsilon, \delta$}; Accuracy parameter $\beta$.
\\\textbf{Output.} Noisy quantile estimates \scalebox{0.9}{$\Ans=\{\a_1, \cdots, \a_m\}$}
\\
\textbf{Initialization.}
\begin{enumerate}[nosep, label=\arabic*.]
\item Set slicing parameter $h = \frac{12}{\epsilon}\ln (\frac{m|D|}{\beta})$
\item Set cont.~count.~param.~$w = \frac{24}{\epsilon}\log (m)\ln(\frac{2m}{\delta})$
\end{enumerate}
\textbf{Assumption.} All quantiles are $\geq \frac{2}{n}(w + h + 1)$ apart \\
\textbf{Functionality.}
\begin{enumerate}[nosep, label=\arabic*.]
\setItemnumber{3}
\item \textbf{For} $i \in [m]$
    \item \hspace{0.5cm} $c_i = \lfloor nq_i\rfloor-h$,  $\d_i = \lfloor nq_i\rfloor+h$
    \item \textbf{End For}
\item \scalebox{0.7}{$\{\Xs[c_i-w,\d_i+w]\}_{i=1}^m$}\scalebox{0.7}{$\gets\textsf{PartialSort}(X, \{[a_i-w, b_i+w]\}_{i=1}^m)$} \\\textcolor{blue}{$\rhd$} Only ranks in \scalebox{0.8}{$[c_i-w, d_i+w]$} need to be in order
\item Sample $\eta_{1}, \cdots, \eta_{m} \leftarrow \frac{w}{2}+\CC_{\varepsilon/2}(m)$
\item \textbf{If} $\calA$ is passive, 
\item \hspace{0.5cm} Sample $\eta^\calA_{1}, \cdots, \eta^\calA_{m} \leftarrow \frac{w}{2}+\CC_{\varepsilon/2}(m)$
\item \textbf{Else}  
\item \hspace{0.5cm} $\ad$ $\calA$ sends $\eta^\calA_{1}, \cdots, \eta^\calA_{m} \in \Z^m$ or $\textsc{ABORT}$
\item \hspace{0.5cm} \textbf{If} \textsc{ABORT}, then halt
\item \textbf{End If}
\item \textbf{For} $i \in [m]$
\item \hspace{0.5cm} \textbf{If} $\eta^\calA_i \not \in [0, w]$, then halt
\item \textbf{End For}
\item \textbf{For} $i \in [m]$
    \item \hspace{0.5cm} $\Delta_i = \eta_i -\eta^\calA_i$, $c_i'=c_i+\Delta_i$, $\d_i'=\d_i+\Delta_i$
    \item \hspace{0.5cm} $\a_i=\FEM(\Xs[c'_i, \d'_i], 0.5, \frac{\epsilon}{6})$ 
    \item \textbf{End For}
    \item Output $\Ans=\{\a_1, \cdots, \a_m\} $($\ad$ send  $\Ans$ to $\calA$)
\end{enumerate}
\end{mybox2}

\begin{minipage}{\columnwidth}
\captionof{figure}{Ideal Functionality: $\FSEM$}\label{box}
\end{minipage}



    \noindent\textbf{Distinctions from \cite{imola2025differentiallyprivatequantilessmaller}.} While this strategy is suitable in the central model, it is ill-suited to our setting,  since it requires privatizing the ranks and necessitates expensive secure exponentiation.
To make \SEM~amenable to MPC protocols, we introduce two key technical distinctions:
\squishl
\item  The sampled noise must be bounded and non-negative. 
\item In \SEM, the slices are perturbed as:
\begin{gather*} 
\tilde{S}_i=\Xs[c_i +\Delta_i, \d_i+\Delta_i] \qquad 
\Delta_i=\eta_i-\eta^\calA_i\\
\boldsymbol{\eta} \leftarrow \CC_\epsilon(m)\qquad  \eta_i^\mathcal{A} \in \mathbb{Z}^+_{\leq w}
\end{gather*}
Specifically, the randomized shift $\Delta_i$ is the \textit{difference} between two noise samples -- $\eta_i$ is generated from $\CC$ while $\eta^\calA_i$ is contributed to by $\calA$. Here, $w$ is a utility parameter determined by $\CC_\varepsilon$. 
\squishe
These changes let us implement the mechanism under MPC by first computing the original slices $S_i$ using the (public) target ranks $r_{q_i}$ and then \textit{securely shifting} them to generate the perturbed slices $\tilde{S}_i$--we elaborate on this in Sec.~\ref{sec:prot:slicing}. 
\\\noindent\textbf{Privacy Analysis.}
 Proving the privacy of $\FSEM$~entails showing that the view of the adversary $ \textsc{View}^{\calA}_{\FSEM}(X, Q, \epsilon) = (\boldsymbol{\eta}^\calA, \Ans)$ is 
$(\epsilon, \delta)$-DP, where $\boldsymbol{\eta}^{\calA}$
  represents the noise contributed by the adversary and $\Ans$ is the final output of the functionality.
\begin{theorem} The ideal functionality $\FSEM$ satisfies $(\epsilon, \delta e^{\varepsilon})$-DP. \label{thm:privacy:SEM}
\end{theorem}

\noindent\textbf{Utility Analysis.} We show that the utility of $\FSEM$ is the same as that of~\cite{imola2025differentiallyprivatequantilessmaller} formalized as:
\begin{theorem}\label{thm:slice-utility}\cite[Theorem 5.1]{imola2025differentiallyprivatequantilessmaller}
    With probability at least $1-2\beta$, all returned estimates of $\FSEM$ satisfy $\textnormal{err}_{X,q_i}(v_i) \leq \frac{12\ln(|D|m/\beta)}{\epsilon } + \frac{24\log (m) \ln(2m / {\beta})}{\epsilon }$.
\end{theorem}

\subsection{Bucketing}\label{sec:IF:bucketing}

We assume that the target quantiles are sufficiently spaced.
The \Buck~mechanism inputs the private dataset $X$ and the estimates of the bounding quantiles $\Ans'= \{\a_1', \cdots, \a'_{2m}\}$ where $(\a'_{2i-1}, \a'_{2i})$ are the estimates for $q_i\pm \alpha$, where $\alpha$ is an accuracy parameter precisely defined later. Now observe that these estimates essentially  partition the input data domain 
$\B=[0,\d]$ into $2m+1$ buckets (or intervals) of the form $\{[0,\a_1'), (\a_1', \a_2'], \cdots, (\a_{2m}',\d]\}$. 
The goal of the mechanism is to compute a differentially private histogram over these buckets--that is, the (noisy) count of the number of data points in $X$
that fall into each bucket. While a straightforward approach would be to apply the Laplace mechanism  (using a truncated and shifted Laplace distribution to ensure non-negative noise), we instead opt for a noise distribution better suited to our problem based on the following insight.

The prefix sums of the histogram, i.e., the cumulative counts $\sum_{j=1}^{2i-1}|\Bu_j|$ is a noisy estimate for the rank of the bounding quantiles $\a'_{2i-1}$. Now recall, that the bucket $B_{2i}$ is used to estimate the target $q_i$ and the rank information directly affects the utility scores used by \SEM. Therefore, ensuring the accuracy of these cumulative sums is essential for the utility of downstream quantile estimation. To illustrate how the prefix sums change under substitution, consider adjacent datasets $X$ and $X'$, where a data point originally in bucket $\Bu_t$ is replaced by one in $\Bu_k$, for $1 \leq k < t \leq 2m+1$. As a result, the sums up to bucket $\Bu'_k$ remain unchanged, the sums from $\Bu'_{k+1}$ to $\Bu'_t$ each increase by 1, and the sums beyond $\Bu'_t$ remain unaffected. 
This structured shift is precisely the setting in which the $\CC$~mechanism achieves substantial utility gains—a key observation we leverage here. 
\begin{mybox2}{Ideal Functionality: $\FBuc$}
\textbf{Input.} Private dataset \scalebox{0.8}{$X \in \B^n$}; Boundary quantile estimates \scalebox{0.8}{$V'=\{v'_1, \cdots, v'_{2m}\} \in [\B]^{2m}$}; Privacy parameters $\epsilon, \delta$.\\
\textbf{Output.} Noisy histogram over the buckets $\mathbf{cnt}$\\
\textbf{Initialization.}
\begin{enumerate}[nosep, label=\arabic*.]
\item Define $\tau=\lceil\frac{6}{\epsilon} \log(m) \ln(\tfrac{2m}{\delta})\rceil$
\end{enumerate}
\textbf{Functionality.}
\begin{enumerate}[nosep, label=\arabic*.]
\setItemnumber{2}
\item Sample $\gamma_1, \cdots, \gamma_{2m+1} \leftarrow \CCd_\epsilon (2m+1)$
\item $\ad$ Send $\sum_{i=1}^{2m+1}{\gamma_i}$ to $\calA$
\item \textbf{If} $\calA$ is passive
\item \hspace{0.2cm} Sample $\gamma_1, \cdots, \gamma_{2m+1} \leftarrow \CCd_\epsilon(2m+1)$
\item \textbf{Else} 
\item \hspace{0.2cm} $\ad$ $\calA$ sends $\gamma_1, \cdots, \gamma_{2m+1} \in \mathbb{Z}^{2m+1}$ or \textsc{ABORT}
\item \hspace{0.2cm} \textbf{If} \textsc{ABORT}, then halt
\item \textbf{End If}
\item \textbf{For} $i \in [2m+1]$
\item \hspace{0.2cm} \textbf{If} $|2i\tau-\sum_{j=1}^i \gamma^\calA_j|> \tau$ or $\gamma_i^\calA \notin [0, 4\tau]$
\item \hspace{0.8cm} Halt
\item \hspace{0.2cm} \textbf{EndIf}
\item \textbf{End For}
\item \textbf{For} $i \in [2m+1]$
\item \hspace{0.2cm} Form buckets \scalebox{0.9}{$\Bu_i=\{x \in X| h'_{i-1} \leq x < h'_{i}\}$ } \\\textcolor{white}{\hspace{2.5cm}}with $h_0' = 0$ and $h_{2m+1}'=\d$
\item \hspace{0.2cm} $cnt_i=|\Bu_i|+\gamma_i+\gamma^\calA_i$
\item \textbf{End For}
\item Output \scalebox{0.85}{$\mathbf{cnt} =\{cnt_1,\cdots, cnt_{2m+1}\}$} ($\ad$ send $\mathbf{cnt}$ to \scalebox{0.85}{$\calA)$}
\end{enumerate}
\end{mybox2}

 \noindent\begin{minipage}{\columnwidth}
\captionof{figure}{Ideal Functionality: $\FBuc$}\label{fig:IF:Bucketing}
\end{minipage}
However, our setting introduces two challenges that deviate from standard $\CC$: (1) the number of dummy users $\gamma_i$ assigned to each bucket must be a non-negative integer, and (2) their \emph{partial sums} $\sum_{i=1}^j \gamma_i$ must have low error to ensure accurate prefix sums.  We achieve this by setting $\gamma_i = 2\tau + \eta_i - \eta_{i-1}$, where $\eta_i$ is generated from $\CC$, and $\tau$ is a sufficiently large integer that guarantees the difference will be non-negative. The servers then subtract a multiple of $2\tau$ from any partial sum to de-bias the estimate. We refer to this new mechanism as $\CCd$ (for \emph{consecutive differences} of $\CC$; details in App.~\ref{app:cont-count}).

As before, two copies of noise $\gamma_i, \gamma_i^\calA$ are generated, and the adversary $\calA$ may generate their noise arbitrarily. Additionally, $\calA$ is allowed to learn the total amount of noise added by the functionality across all $2m+1$ buckets, since this is the total number of dummy users added.
\\\noindent\textbf{Privacy Analysis.} Since the dataset size $n$ is known, the leaked sum $\sum_{i=1}^{2m+1} \gamma_i$ reveals no more than $\mathbf{cnt}$, so privacy follows from that of $\CCd$.
\begin{theorem} The ideal functionality $\FBuc$ is $(\epsilon,\delta)$-DP. \label{thm:privacy:bucketing}
\end{theorem}

\noindent\textbf{Utility Analysis.}
By the utility guarantees of $\CCd$, $\boldsymbol{\gamma}$ satisfies a linear function: $|2j\tau-\sum_{i=1}^j \gamma^\calA_i|> \tau$ where $\tau$ is a parameter determined by the noise distribution.
Hence, the functionality aborts if the $\calA's$ contribution doesn't satisfy this. This enables us to give the following utility theorem even when $\calA$ is active.
\begin{theorem}\label{thm:bucket-util}
    If $\FBuc$ does not halt, then the returned vector of counts satisfies, for all $1 \leq j \leq m$, 
    \begin{gather*}
        \left \lvert - 4\tau j + \sum_{i=1}^j cnt_i - \sum_{i=1}^j \lvert B_i \rvert \right \rvert \leq \frac{12}{\epsilon} \log (m) \log (\tfrac{2m}{\beta}),
    \end{gather*}
    and $cnt_i - |B_i| \in [0, 8\tau]$. Further, assuming both players follow protocol, $\FBuc$ halts with probability at most $1-2\beta$.
\end{theorem} 
Thus, the noisy partial sums $\sum_{i=1}^j cnt_i$, closely match the partial sums $\sum_{i=1}^j |B_i|$ required to re-normalize the quantiles.
\subsection{Putting it all together}\label{sec:everything}
We are now ready to describe the full functionality of \name.
Recall that \name~performs quantile estimation in two phases to improve scalability with respect to the dataset size $n$. Given a set of target quantiles $Q = {q_1, \ldots, q_m}$, the goal of the first phase is to estimate a set of bounding quantiles $Q'=\{q_i \pm \alpha| q_i \in Q\}$, for some accuracy parameter $\alpha \in (0,1)$. These bounding quantiles are used only to efficiently \textit{isolate} small, individual slices of $X$—that is, data points with ranks in a neighborhood around each $q_i$—which are relevant for estimating the target quantiles in the second phase. Hence, these auxiliary estimates need not be too accurate.
To reduce computational overhead, this first phase is run on a subsample $\Xsam \subset X$ of size $k \ll n$, drawn uniformly at random from the full dataset $X$ of size $n$. The auxiliary estimates are computed via
$V' = \FSEM(\Xsam, Q', \epsilon_1)\footnote{We omit some input parameters for notational brevity.}$
using a portion $\epsilon_1$ of the total privacy budget. The parameter $\alpha$ is chosen to account for the sampling error introduced by working on $\Xsam$ instead of $X$ as well as error of $\FSEM$.


The auxiliary estimates partition the input domain into $2m + 1$ buckets. The functionality then assigns each data point in $X$ to its corresponding bucket and invokes $\FBuc$ to obtain a noisy histogram $\mathbf{cnt}$ over these buckets, consuming a privacy budget of $\epsilon_2$. In the actual secure protocol, this histogram is computed by adding dummy records $R$, which is accounted for in Step 9 of Fig.~\ref{fig:IF:name}. In the final phase, the target quantiles are estimated from the individual buckets $\Bu_{2i}=\{x\in X\cup R| \a'_{2i-1} \leq x < \a'_{2i}\}$ using the final privacy budget $\epsilon_3$. Since the ranks of elements in $\Bu_{2i}$ do not correspond to the ranks in the entire dataset, the target quantiles must be \textit{re-normalized} to ensure correctness. This is done using prefix sums of the noisy histogram $\mathbf{cnt}$, which estimate the (noisy) global rank of the bounding quantile $\a'_{2i-1}$. The re-normalized quantile is then given by:
\[\hat{q}_i =  \min\Bigg\{1, \max\bigg\{0, \bigg\lfloor \frac{q_i n +2i\tau- \sum_{j=1}^{2i-1}cnt_i}{cnt_{2i}}\bigg\rfloor\bigg\}\Bigg\} \]
where $\tau$ is an accuracy parameter determined by $\CC$. Again, since the quantity of interest is prefix sums, using $\CC$ ensures the highest accuracy for the prefix sum estimates.
\begin{mybox2}{Ideal Functionality: $\F_\name$}
\textbf{Input.} 
Private dataset $X \in \B^n$; Queried quantiles $Q = \{q_1, \cdots, q_m\} \in (0,1)^m$; Privacy budgets $\varepsilon_1$ for first quantile estimation, $\varepsilon_{2}$ for bucketing, and $\varepsilon_{3}$ for final quantile estimation; Privacy parameter $\delta$, accuracy parameter $\beta$. 
\\\textbf{Output.} Noisy quantile estimates \scalebox{0.9}{$V=\{\a_1, \cdots, \a_m\}$}\\
\textbf{Assumption.} 
Quantiles are spaced at least $\frac{24}{n \epsilon_3} \ln(\frac{d m}{\beta}) + \frac{48}{n \epsilon_3}\log(m) \ln(\frac{2m}{\delta})$ apart.\\
\textbf{Initialization.}
\begin{enumerate}[nosep, label=\arabic*.] 
\item Define $k = (nm)^{2/3} (2 \ln(\frac{2}{\beta}))^{1/3}$\footnote{The exact value of $k$ is defined in Theorem~\ref{thm:runtime}.}, $\alpha_1 = \frac{12 \ln(|D|m/\beta)}{k\epsilon_1} + \frac{24 \log(m) \ln(2m/\beta)}{k\epsilon_1}$, $\alpha_2 = \sqrt{\frac{\ln(2/\beta)}{2k}}$, $\alpha = \alpha_1 + \alpha_2$, and $\tau=\frac{6}{\epsilon_2} \log(m) \ln(\tfrac{2m}{\delta})$.
\item Merge quantiles closer than $4\alpha_1 + 2\alpha_2$ into sets; obtain quantile sets $Q_1, \ldots, Q_m$ partitioning $Q$.
\end{enumerate}
\textbf{Functionality.}
\begin{enumerate}[nosep, label=\arabic*.]
\setItemnumber{3}
\item $\Xsam \sim$ \textsf{Subsample}$(X, k)$ 
\item $Q'=\bigcup_{i=1}^m\{\min_{q \in Q_i}(q-\alpha), \max_{q \in Q_i}(q + \alpha) \}$
\item $V'=\FSEM(\Xsam, Q', \epsilon_{1}, \delta)$
\item cnt = $\FBuc(X, V', \epsilon_2, \delta)$
\item \textbf{For} $i \in [2m+1]$
\item \hspace{0.2cm} Form bucket $\Bu_i = \{x \in X| \a'_{i-1} \leq x < \a'_{i}\}$
\item \hspace{0.2cm} Add $cnt_i-|\Bu_i|$ dummy records in $\Bu_i$ with value $\a'_{i-1}$
\item \textbf{End For} 
\item \textbf{For} $i \in [m]$
\item \hspace{0.2cm} $\hat{Q}_i =  \underset{q \in Q_i}{\bigcup}\{\min\{1, \frac{\max\{0, q n + 8 \tau i - \sum_{j=1}^{2i-1}cnt_j\}}{cnt_{2i}}\}\}$
\item \hspace{0.2cm} $V_i=\FSEM(\Bu_{2i},\hat{Q}_i,\epsilon_{3}, \delta)$
\item \textbf{End For} 
\item Output $V=V_1 \cup \cdots \cup V_m$ ($\ad$ send $V$ to $\calA$)
\end{enumerate}
\end{mybox2}
 \noindent\begin{minipage}{\columnwidth}
\captionof{figure}{Ideal Functionality: $\F_{\name}$}\label{fig:IF:name}
\end{minipage}
\\\noindent\textbf{Note.} $\F_\name$~invokes $\FSEM$~twice. In the first phase, a single instance of $\FSEM$ estimates all $2m$ bounding quantiles from the subsample $\Xsam$. In the second phase, $m$ independent instances of $\FSEM$ are run—one on each bucket $\Bu_{2i}$ to estimate the corresponding target quantile $q_i$. Both phases benefit from parallel composition and operate on smaller datasets—either the subsample or individual buckets—enhancing both utility and scalability with  $n$.

To preserve privacy, the bounding quantile intervals must not overlap, i.e., intervals $(q_i-\alpha, q_i+\alpha)$ and $(q_j-\alpha, q_j+\alpha)$ cannot intersect; otherwise data points could contribute to multiple estimates. To prevent this, quantiles within roughly $\alpha$ of each other are merged into sets $Q_1, \dots, Q_m$. Each set is then processed as a batch: one pair of bounding quantiles define a shared bucket, and the final call to $\FSEM$ estimates all quantiles in that set from the same bucket.
\\\noindent\textbf{Privacy Analysis.}
Privacy of $\F_{\name}$ follows directly from adaptive composition~\cite{dwork2010differential}:
\begin{theorem}
    Assuming $\varepsilon_1,\varepsilon_2,\varepsilon_3 = O(1)$ the functionality $\F_{\name}$ satisfies $( \epsilon_1 + \epsilon_2 + \epsilon_3, O(\delta))$-DP.\label{thm:name}
\end{theorem}
As with $\FSEM$, privacy of $\F_{\name}$ holds even for a malicious server---the noise injected by one honest server is enough to ensure DP.

\noindent\textbf{Utility Analysis.} The utility of $\F_{\name}$ follows from analyzing the noise injected by $\FBuc$ and by the second call to $\FSEM$, which is responsible for the final estimates.

\begin{theorem}\label{thm:main-util}
    For any input $X \in D^n$, with probability at least $1-6\beta$, all quantile estimates $v_1, \ldots, v_m$ by $\F_{\name}$ satisfy
    \[
        \mathsf{err}_{X,q_i}(v_i) \leq O\left(\frac{\log(\frac{|D|m}\beta)}{ \epsilon_3} + \frac{\log (m) \log (\frac{m}{\beta})}{ \min(\epsilon_2, \epsilon_3)}\right).
    \]
\end{theorem}
We note that $\epsilon_1$ does not impact utility. Intuitively, the initial subsampling step is simply an alternative method for forming slices in the final quantile computation, chosen to reduce the number of secure comparisons. Thus, $\epsilon_1$ influences only runtime (see Sec.\ref{sec:run}). Guidelines for setting hyperparameters, including $\epsilon_1$, are provided in App.~\ref{sec:set-hyperp}.

\noindent\textbf{Quantile Gap Assumption.} 
$\F_{\name}$ has roughly the same gap requirement as $\FSEM$; both assumptions are mild and work out to $\leq 0.01$ on datasets with $\approx 10^6$ elements. 


\section{\name: Protocols}\label{sec:protocols}
In this section, we describe the protocols that implementation the above discussed ideal functionalities under MPC. As before, we start with a bottom up approach.

\subsection{Secure SlicingEM }\label{sec:prot:slicing}
\subsubsection{Slicing}
As described in Sec.~\ref{sec:IF:slicing},  each slice is a contiguous segment of the sorted dataset of the form $S_i = \Xs[c_i, d_i]$ and they are shifted by $\Delta_i = \eta^0_i - \eta^1_i$, where $\eta_0, \eta_1 \in \mathbb{Z}^{+}_{<w}$ are noise samples.
\footnote{At least one of $\eta_0$ or $\eta_1$ is drawn from the $\CC$ mechanism, but since this detail is not essential to the MPC implementation, we omit it here.}
We first, discuss a cleartext algorithm for implementing the above shift and then outline how to translate this to an MPC protocol.  
Consider an extended slice of size $2(h+w)$, defined as $\bar{S}_i=\Xs[c_i-w,d_i+w]$, and perform the following operations:
\squishl
\item Sample $\eta_0 \in \mathbb{Z}^+_{<w}$ and add $d$ to the first $\eta_0$ elements of $\bar{S}_i$. 
\item Sample  $\eta_1 \in \mathbb{Z}^+_{<w}$ and subtracts $d$ from the last $\eta_1$ elements of $\bar{S}_i$. 
\item Sort the slice perturbed as above; let   $\hat{S}_i$ denote the resulting slice.
\squishe

\begin{mybox}{Protocol: Secure SlicingEM $\Pi_{\SEM}$}
\textbf{Input:} 
Secret-shared private dataset $\langle X \rangle$ where $X\in \B^n$; Queried quantiles $Q = \{q_1, \cdots, q_m\} \in (0,1)^m$; Privacy parameter $\epsilon, \delta$;
  \\       \textbf{Output:} Private quantile estimates \scalebox{0.9}{$V=\{v_1, \cdots, v_m\}$}
\\\textbf{Initialization:}
\begin{enumerate}[nosep, label=\arabic*.] 
\item Set slicing parameter $h = \frac{12}{\epsilon}\ln (\frac{m|D|}\beta)$
\item Set $\CC$ parameter \scalebox{0.9}{$w = \frac{24}{\epsilon}\log (m)\log(\frac{m}{\delta})$}\end{enumerate}
\textbf{Protocol:}
 Each server $\Ser_b$   \\
\textcolor{blue}{ \hspace{3.5cm}$\rhd$} Masking array generation
 \begin{enumerate}[nosep, label=\arabic*.] 
 \setItemnumber{3}
        \item  $\eta^b_1, \ldots, \eta^b_m \gets \CC(\frac{\epsilon}{2}, m)$ 
        \item Generates an array of size $mw$ such that
        \begin{gather*}  L_b[i] = \begin{cases} (-1)^b\d, & \text{ if } i \leq (\lceil \frac{i}{w} \rceil-1)w + \eta^b_{\lceil \frac{i}{w} \rceil}\\
        0, & \text{otherwise}
        \end{cases}
        \end{gather*}
        \item $\langle L_b\rangle$ =\textsf{Share}$(L_b)$
        \item \textbf{For} $j \in \{0,1\}$ 
        \item \hspace{0.07cm}\textbf{For} $i \in [mw]$

        \item \hspace{0.25cm} \textbf{If} \scalebox{0.8}{$\Rec\big(\Mult(\langle L_j[i]\rangle - (-1)^j
\d,\langle L_j[i]\rangle) \big)!=0$}, then halt
        \item  \hspace{0.07cm} \textbf{EndFor}
        \item\textbf{EndFor}
        \item $\langle \hat{L}_b\rangle = \F_{\textsf{Sort}\downarrow}(\langle L_b \rangle)$ \textcolor{blue}{$\rhd$} Descending order
\item \textbf{For} $i \in [m]$
    \item \hspace{0.5cm} $c_i = \lfloor nq_i\rfloor-h-w$,  $d_i = \lfloor nq_i\rfloor+h+w$
\item $\langle \hat{X}\rangle\gets\F_\textsf{PartialSort}(\langle X\rangle, \{[c_i, d_i]\}_{i=1}^m)$
\item \textbf{For} $i \in [m]$
\item \hspace{0.2cm}    \textbf{For }$j \in [w]$
   \item $ \hspace{0.4cm}\langle \hat{X}[c_i+j]\rangle=\langle \hat{X}[c_i+j]\rangle +\langle L_0[(i-1)m+j]\rangle $
\item $ \hspace{0.4cm}\langle \hat{X}[d_i-j]\rangle=\langle \hat{X}[d_i-j]\rangle+\langle L_1[(i-1)m+j]\rangle$
\item \hspace{0.2cm}
   \textbf{End For}
 \item \textbf{End For}
   \item \textbf{For } $i \in [m]$
   \item \hspace{0.5cm} $\langle X_i\rangle =\F_\textsf{Sort}( \langle \hat{X}[c_i:d_i]\rangle)$
   \item \hspace{0.5cm} $v_i=\Pi_\textsf{EM}(\langle X_i[w:-w]\rangle, q_i, \frac{\epsilon}{6}) $
\item \textbf{End For}
\item \textbf{Return} $V=\{v_1, \cdots, v_m\}$
\end{enumerate}
\end{mybox}

 \noindent\begin{minipage}{\columnwidth}
\captionof{figure}{Protocol  $\Pi_{\SEM}$}\label{fig:prot:SEM}
\end{minipage}

We claim that $\tilde{S}_i[w:-w]$ gives the desired shifted slice.  Given the input domain $\B=[0,\d]$, adding (or subtracting) $d$ acts as a mask that pushes certain elements outside the original input range. Sorting the perturbed slice then shifts these masked elements toward the edges of the array. We formalize this as: 
\begin{lemma}
The above transformation is equivalent to shifting $S_i$ by  $\Delta_i = \eta^0_i - \eta^1_i$, i.e., $\hat{S}_i[w:-w]=\Xs[c_i+\Delta_i:d_i +\Delta_i]$. (proof is in App. \ref{app:lemma}). \label{lemma}
\end{lemma}

Now, let us turn to how to implement this under MPC. First both the servers sample noise $\boldsymbol{\eta}\gets \CC_{\epsilon}(m)$. Next,  for each slice $S_i$, the servers $\Ser_{b}, b \in \{0,1\}$ create a secret-shared masking array $\langle L^b_i \rangle$ of size $w$ as follows.  Server $\Ser_0$ sets the first $\eta^0_i$ elements of $L^0_i$ to $\d$ and fills the rest with zeros while server $\Ser_1$ sets the last $\eta^1_i$ elements of $L^b_i$ to $-\d$ and fills the rest with zeros. They then exchange shares of the corresponding arrays. The above shift can now be implemented by simply performing element-wise addition: $\langle \bar{S_i}[:w]\rangle + \langle L^0_i \rangle$ and $\langle \bar{S_i}[-w:]\rangle +\langle L^1_i \rangle$. Securely sorting this perturbed slice yields the desired shifted slice, as previously described.

A caveat is that, since we operate in the malicious model, we must ensure that a corrupt server doesn't construct malformed masking arrays. We do this as follows: \squishl
    \item For all $j \in [w]$, securely verify that each element satisfies $(L^b_i[j]-1^{1-b})L_i=0$. This ensures that $L^b_i[j] \in \{0,\d,-\d\}$.
    \item Securely sort the arrays in descending order - this ensures that the non-zero values are contiguous at the beginning (for $\d$) or at the end (for $-\d$) of the array, and 0 elsewhere.
\squishe

\subsubsection{Exponential Mechanism}
The $\Pi_{\EM}$ protocol (Fig. \ref{fig:prot:EM} in App. \ref{app:func}) takes as input a secret-shared array $\langle \Xs \rangle = \{\langle x_{(1)}\rangle, \ldots, \langle x_{(n)}\rangle \}$ in \textit{sorted} order and outputs a differentially private estimate of the target quantile.  For an efficient MPC implementation, we make two key observations.
First, as discussed in Sec.~\ref{sec:key ideas}, \EM~needs to assign selection probabilities to only $n+1$ intervals of the form: 
\begin{gather*} 
\I_0 = [0, x_{(1)}], \quad 
\I_1 = [x_{(1)}, x_{(2)}], \quad 
\ldots, \quad 
\I_n = [x_{(n)}, d],
\end{gather*} where the input domain is $\B=[0,d]$. Thus, we completely avoid dependence on the size of the input domain $|\B|$ which could be quite large. 

Second,  since the intervals are in sorted order, the corresponding utility score reduces to: 
\begin{align*}
\u_q(X,\I_i)& =-\big||\{x \in X|x < y\}| - \rq\big|
=-|i-\rq|
\end{align*}
where $ y \in \I_i$. Hence, as long as the target quantile (equivalently, its rank $\rq$) is public, the utility scores and the subsequent exponentiation can be computed \textit{in the clear}. 
To this end, recall that the secure shifting mechanism discussed above ensures that the queried quantiles remain public. Moreover, the re-normalization required during the second invocation of $\Pi_{\SEM}$ can also be performed using the (public) noisy histogram obtained from the preceding bucketing phase. As a result, the interval weights can be computed \textit{locally} by the servers
 without any costly secure exponentiation:
\begin{gather*}
\langle wt_i \rangle = \exp\left(-\frac{\epsilon}{2}|i - r_q|\right) \cdot \left(\langle \Xs[i+1] \rangle - \langle \Xs[i] \rangle\right).
\end{gather*}
Next, using these weights, $\Pi_{\EM}$ selects an interval via inverse transform sampling~\cite{Kerschbaum}. This involves drawing a uniform random sample $t \in [0, 1]$ and finding the smallest index $j$:
\begin{gather*}
\sum_{i=0}^{j-1} \Pr[\EM_{u}(X, q,\epsilon) = \I_i] \leq t < \sum_{i=0}^{j} \Pr[\EM_\u(X, q, \epsilon) = \I_i].
\end{gather*}
We implement this search securely using a linear scan.

To further reduce MPC complexity, we avoid secure divisions by working with unnormalized weights. Specifically, the uniform sample $t$ is scaled by the total weight:
$N = \sum_{i=0}^{n} wt_i$. The final challenge is to select an element uniformly at random from within the selected secret-shared interval $\I_j = [u_j, u_j + l_j]$. This is done by securely computing $ u_i + \lfloor Nt \cdot l_i \rfloor$ (\textsf{SecureIntervalSearch}, see Fig. \ref{fig:prot:interval})
which gives us the private estimate for $q$.
\subsection{Secure Bucketing}\label{sec:protocol:bucketing}
The output of this step is a set of $2m+1$ buckets containing secret-shared values, where the histogram over these buckets satisfies DP. The primary goal is to securely generate the required number of dummy records and populate the buckets without compromising the utility of the protocol. These correspond to Steps 4-19 in Fig.~\ref{fig:prot:name}. To achieve this, each server $\Ser_b$ samples noise independently $\boldsymbol{\gamma}^b\gets \CCd_{\epsilon_2}(2m+1)$ where $\gamma^b_i$ is the noise sampled for the $i$-th bucket. Each server then creates a number of dummy records equal to the total sampled noise and secret-shares them with the other party.

To ensure that dummy records preserve privacy without degrading utility, we enforce two constraints. First, the value of the dummy records corresponding to bucket $\Bu = [c', d']$ is set to be $c'$. This ensures that, after sorting, all dummy records for a given bucket accumulate at one side of the interval.
Second,  without appropriate checks, a malicious server could inject an arbitrary number of dummy records, thereby distorting the histogram. To prevent this, we impose a bound on the cumulative noise introduced. Specifically, the noise sampled from $\CCd_{\epsilon_2}(m)$ must satisfy a linear growth condition  $|2i\tau-\sum_{j=1}^i \gamma_i|\leq \tau$, for some utility parameter $\tau$ of the noise distribution.  

These two checks are securely enforced on each dummy record ($\Pi_\textsf{Dummy}$, Fig. \ref{fig:prot:dummy}). Once verified, the servers merge the dummy records with the real data and perform a secure shuffle, ensuring that the dummy records are indistinguishable from the real ones. Finally, using the estimates of the bounding quantiles, the servers securely assign all records into the corresponding $2m+1$ buckets.

The full protocol is given in Fig. \ref{fig:prot:name}.


\begin{mybox}{Protocol: $\Pi_{\name}$}
\textbf{Parameters:}\\
Private dataset $X \in \B^n$; queried quantiles $Q = \{q_1, \cdots, q_m\} \in (0,1)^m$; privacy budgets $\varepsilon_1, \varepsilon_{2},\varepsilon_{3}$; privacy parameter $\delta$, accuracy parameter $\beta$. 
\\\textbf{Output:} Private quantile estimates \scalebox{0.9}{$V=\{v_1, \cdots, v_m\}$}
\\
\textbf{Initialization.}
\begin{enumerate}[nosep, label=\arabic*.] 
\item Define $k = (nm)^{2/3} (2 \ln(\frac{2}{\beta}))^{1/3}$, $\alpha_1 = \frac{12 \ln(|D|m/\beta)}{k\epsilon_1} + \frac{24 \log(m) \ln(2m/\beta)}{k\epsilon_1}$, $\alpha_2 = \sqrt{\frac{\ln(2/\beta)}{2k}}$, $\alpha = \alpha_1 + \alpha_2$, and $\tau=\frac{6}{\epsilon_2} \log(m) \ln(\tfrac{2m}{\delta})$.
\item Merge quantiles closer than $4\alpha_1 + 2\alpha_2$ into sets; obtain quantile sets $Q_1, \ldots, Q_m$ partitioning $Q$.
\end{enumerate}
Each server $\Ser_b$
\begin{enumerate}[nosep, label=\arabic*.] 
\setItemnumber{3}
\item $\langle \Xsam \rangle =\Pi_{\textsf{Sample}}(\langle X \rangle, k)$
\item  $Q'=\{\min_{q \in Q_i} q-\alpha, \max_{q \in Q_i} q + \alpha : 1 \leq i \leq m\}$
\item $V'=\Pi_{\SEM}(\langle \Xsam\rangle , Q', \epsilon_1, \delta)$
\\ \textcolor{blue}{\hspace{2.5cm}$\rhd$} Creating dummy records
\item Sample $\boldsymbol{\gamma} \gets \CCd_{\epsilon_2}(2m+1)$
\item $v'_0=0$
\item $R_b=\{\}$
\item \textbf{For } $i \in [2m+1]$
\item \hspace{0.2cm} $R_b.\text{Append}(\gamma_i,v'_{i-1})$ \\\textcolor{blue}{\hspace{1cm}$\rhd$} Appending $\gamma_i$ records with value $v'_{i-1}$
\item \textbf{End For}
\item \textsf{Share}$(R_b)$ 
\item $\bar{V}=\{v'_0, \cdots, v'_{i-1}\}$
\item \textbf{If} $(!\Pi_\textsf{Dummy}(\langle R_0\rangle, \bar{V}, \tau ) \wedge \Pi_\textsf{Dummy}( \langle R_1\rangle, \bar{V}, \tau)$
\item \hspace{0.5cm} \textsc{ABORT}
\item \textbf{End If}
 \\\textcolor{blue}{\hspace{3.5cm}$\rhd$} Populating buckets
\item $\langle \X_{\textsf{Shuffle}}\rangle = \Fshuffle(\langle \X \rangle \cup \langle R_0\rangle  \cup \langle R_1 \rangle )$
\item \textbf{For } $i \in [n+|R_1|+|R_0|]$  
\item \hspace{0.1cm} Put \scalebox{0.75}{$\langle \X_{\textsf{Shuffle}}[i] \rangle $} in bucket \scalebox{0.75}{$\Bu_j$} s.t. \scalebox{0.75}{$ v'_{i-1}\leq \X_{\textsf{Shuffle}}[i] < v'_{i}$}
\item \textbf{End For}
\item \textbf{For} $i \in [m]$
\item \hspace{0.2cm} \scalebox{0.9}{$\hat{Q}_i =  \{\min\{1, \max\{0, \frac{q n + 8i\tau - \sum_{j=1}^{2i-1}|B_j|}{|B_{2i}|}\}\} : q \in Q_i\}$}
\item \hspace{0.2cm} $V_i=\Pi_\SEM(\langle \Bu_{2i}\rangle ,\hat{Q}_i,\epsilon_{3}, \delta)$
\item \textbf{End For} 
\item Output $V=V_1 \cup \cdots \cup V_m$ 

\end{enumerate}
\end{mybox}

 \noindent\begin{minipage}{\columnwidth}
\captionof{figure}{Protocol  $\Pi_{\name}$}\label{fig:prot:name}
\end{minipage}
\subsection{Security Analysis}
We use the simulation paradigm~\cite{Oded} to prove \name's security guarantees. 
As discussed in Sec. \ref{sec:background}, proving security under the simulation paradigm requires
defining two worlds: the real world,  where the actual protocol is executed by real parties, and an ideal world where an ideal functionality $\mathcal{F}$ receives inputs from all parties and directly computes and returns the output to the designated parties. A simulator $\mathsf{Sim}$ interacts with $\mathcal{A}$ and $\mathcal{F}$ to simulate the view for $\mathcal{A}$.  If $\mathcal{A}$ cannot distinguish whether it is interacting with real parties or with $\mathsf{Sim}$ in the ideal world, then we conclude that the protocol securely realizes $\mathcal{F}$. This guarantees that the protocol leaks no more information than what is explicitly allowed by the ideal functionality. This is formalized as follows.

\begin{theorem}
$\Pi_{\name}$ achieves the ideal functionality in Fig. \ref{fig:IF:name} in the $\F_\textsf{Sort}$ and $\F_\textsf{Shuffle}$-hybrid model  in the presence of an active adversary
corrupting one server and an arbitrary number of clients (proof is in App. \ref{app:security}). \label{thm:security} 
\end{theorem}
Following \cite{HPI}, we next argue that \name~achieves computational
DP against an active adversary corrupting one
of the servers and any number of clients. Thm. \ref{thm:name} shows that a centralized version of the \name~satisfies $(\epsilon_1+\epsilon_2+\epsilon_3,O(\delta))$-DP. The following theorem then follows directly from Thm. \ref{thm:security} and Def. \ref{def:compDP}.
\begin{corollary} 
$\Pi_\name$~satisfies $(\epsilon_1+\epsilon_2+\epsilon_3,O(\delta))$-computational DP against an active adversary corrupting one server and an arbitrary number of clients.
\end{corollary}

\subsection{Implementation Optimizations}\label{sec:opt}
\noindent\textbf{Offline Computation.}
We observe that certain computations can be offloaded to an offline phase, independent of the input dataset. Specifically, the formation and verification of masking arrays can be performed ahead of time since they do not depend on the actual data. Similarly, all calls to the randomness generation function 
$\Frand$ can also be moved offline, reducing the online overhead.

\noindent\textbf{Sorting Implementation.}
Standard secure sorting typically relies on data-oblivious sorting networks. Although sorting networks like AKS achieve optimal asymptotic complexity of 
$O(n\log n)$, their high constant overhead renders them computationally expensive in practice. To improve efficiency, we adopt the shuffle-then-comparison paradigm~\cite{graph1, graph2}. The core idea is to first perform a secure shuffle of the dataset, which randomizes the input order and eliminates any input-dependent access patterns. A standard comparison-based algorithm, such as Quicksort, is then applied. 
Because the input order has been randomized, the outcomes of secure comparisons can be safely revealed, allowing subsequent swaps to be executed in the clear—substantially improving performance.

This paradigm integrates naturally with \name’s design. Since our protocol already requires shuffling to mix dummy and real records for secure bucketing, this enables us to apply comparison sort directly to the inputs for the second invocation of $\Pi_\SEM$. Moreover, this strategy allows us to exploit the performance benefits of partial sorting within \SEM, significantly reducing the number of secure comparisons required. This improvement is formalized as follows:

\begin{theorem}[\cite{Aanders}]\label{coro:partial-sort} The number of comparisons required to partially sort $m$ disjoint intervals, each of length $l$, is $O(n \log(m) + m l \log(n))$ (see App. \ref{app:partial sort} for proof). 
\end{theorem}



\subsection{Runtime Analysis}\label{sec:run}
The concrete performance of \name~is primarily determined by the number of secure comparisons it performs. We therefore analyze the protocol’s comparison complexity below.

 \begin{theorem}\label{thm:runtime}
    By setting the size of the subsample $k$ as 
    \begin{equation*}
        k = (nm)^{2/3}  \max\Big\{\ln (\tfrac{1}{\beta})^{1/3}, \sqrt{\tfrac{\ln(|D|m/\beta) + \log(m)\log(m/\beta)}{\epsilon_1 (nm)^{1/3}}} \Big\},
    \end{equation*}
    the number of secure comparisons in 
    $\F_{\name}$ is bounded by
    $n \log (2m) + O (k \log(m))$.
    \label{thm:secure comp}
\end{theorem}

The $n \log(2m)$ term arises from comparisons in the bucketing step, while the remaining $O(k)$ comparisons come from calls to $\FSEM$. The privacy parameter $\epsilon_1$ affects runtime by influencing bucket sizes for the second $\FSEM$ calls. Since $\epsilon_1$ may be amplified by subsampling (see App. \ref{sec:set-hyperp}), we typically set $k = (nm)^{2/3}\ln(1/\beta)^{1/3}$. The total comparisons by $\FSEM$ remain sublinear until $m = o(\sqrt{n})$, beyond which the analysis of Thm. \ref{thm:runtime} becomes loose. In practice, we found this transition around $m \approx 20$ for $n=10^6$. For larger $m$, a better strategy is to skip the initial $\FSEM$ and $\FBuc$ calls and instead run a single $\FSEM$ over the full dataset with all quantiles.

\section{Experimental Evaluation}\label{sec:eval}

\noindent\textbf{Datasets and Setup.} We consider equally-spaced quantiles
in $[0, 1]$ (for instance, $0.2, 0.4, 0.6, 0.8$ for $m=4$) and do not expect
that uneven spacing would affect runtime or utility. We generated datasets from three distributions over $D=[0,10^9]$: Uniform, a mixture of $\ell=13$ Gaussians, and Zipf—the latter two reflecting common real-world data. The choice of $\ell=13$ makes quantile placement more challenging, since quantiles in $[0,1]$ cut across different Gaussian regions. We also evaluated on two real datasets: credit card balances from a fraud dataset~\cite{ieee-fraud-detection} ($|D|\approx10^5$) and weekly sales from a Walmart dataset~\cite{walmart-recruiting-store-sales-forecasting} ($|D|\approx10^7$), sampling $n$ times with replacement. We use privacy parameter $\delta=10^{-9}$ and dataset size $n=10^6$.

\subsection{Performance Evaluation}

\noindent\textbf{MPC Setup.}
Benchmarks were run on AWS c8g.24xlarge instances using MP-SPDZ~\cite{CCS:Keller20}, with shuffling from~\cite{NDSS:SYBDC24}, RTT $0.01$ ms, $75$ GBit/s bandwidth  and domain size $|\B|=10^9$. Each point is the average of 5 runs. Since MPC reveals nothing beyond outputs, our benchmarks are dataset-agnostic.

\noindent\textbf{Baselines.} We implement two baselines in MPC: (1) SlicingQuantiles ($\SliceEM$)~\cite{imola2025differentiallyprivatequantilessmaller}, and (2) ApproxQuantiles ($\AQ$)~\cite{kaplan2022differentially}, the two state-of-the-art central DP algorithms for estimating multiple quantiles. We omit~\cite{Kerschbaum} because, (1) it assumes the semi-honest threat model, (2) supports only a single quantile (both accuracy and performance degrade linearly with $m$), $(3)$ requires $O(\log |\B|)$ rounds of interaction  while our focus is on non-interactive protocols.

\begin{figure}
	\centering
            \includegraphics[width=0.8\linewidth]{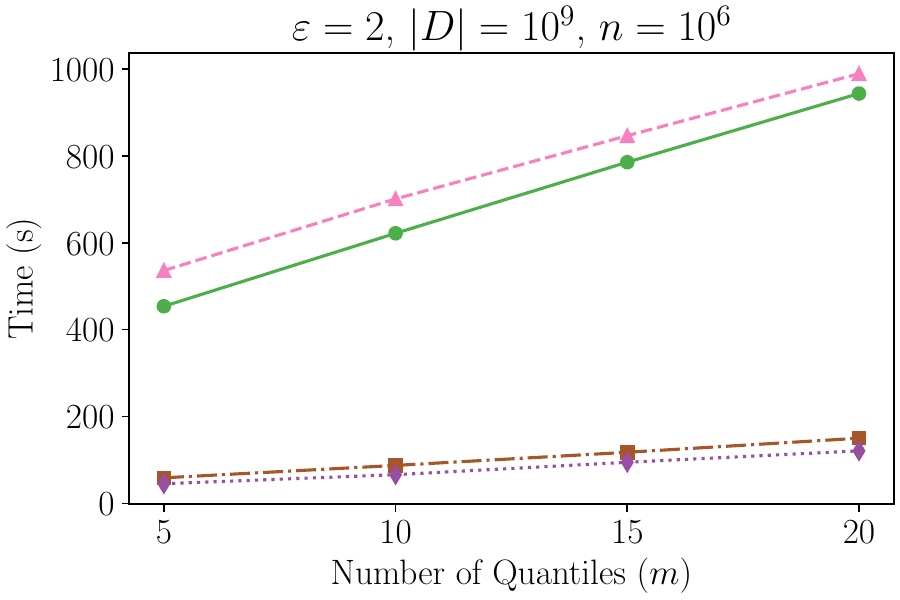}
            \includegraphics[width=0.8\linewidth]{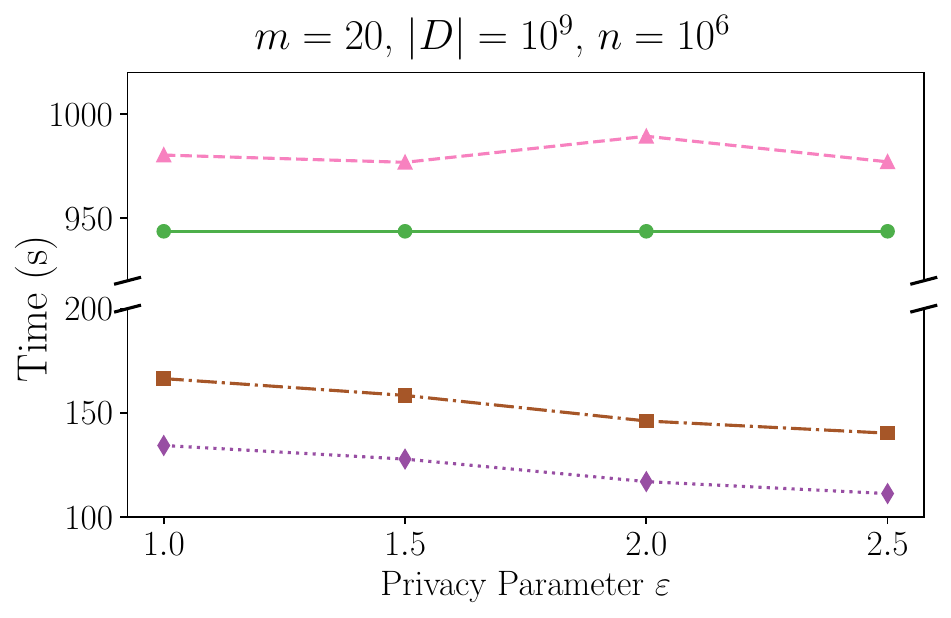}
            \includegraphics[width=0.6\linewidth]{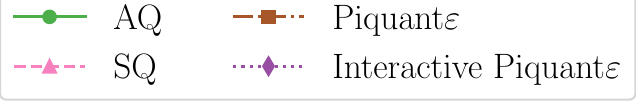}%
    \caption{Runtime plots against varying  $m$ and $\epsilon$ on a dataset of size 1 million.}
    \label{fig:time}
\end{figure}

\begin{figure}
            \centering
            \includegraphics[width=0.8\linewidth]{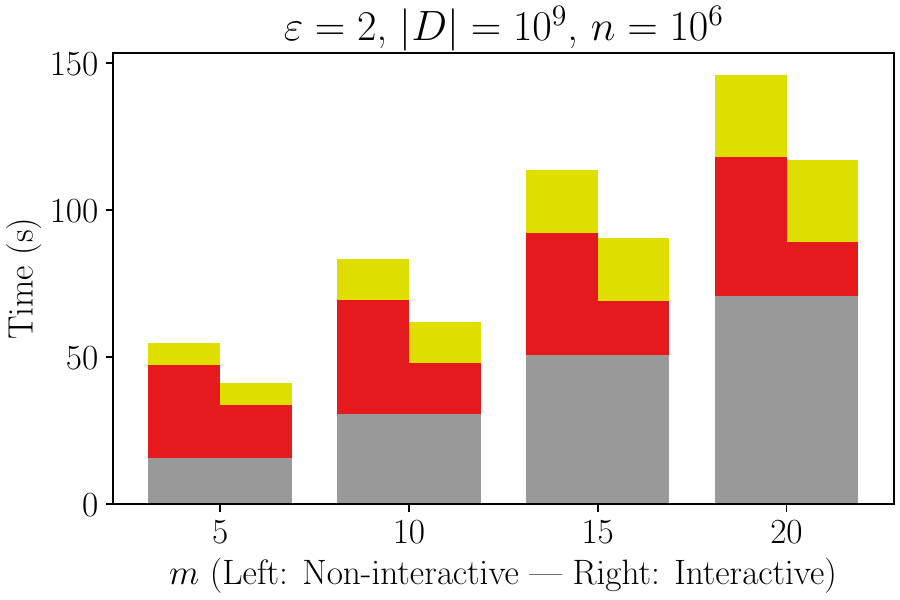}
            \hspace{0.03\textwidth}
            \includegraphics[width=0.8\linewidth]{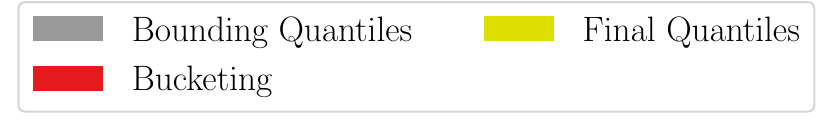}
    \caption{Piquant$\varepsilon$ runtime break down}
    \label{fig:breakdown}
\end{figure}

\noindent\textbf{Results.} On the left, Fig.~\ref{fig:time} reports the performance of \name~with varying $m$, $\epsilon=2$, and dataset size $n=10^6$. \name~is highly efficient: it estimates 5 quantiles in under a minute. It is up to $8.3\times$ faster than $\SliceEM$ and $\sim10\times$ faster than $\AQ$. Runtime grows much more sharply for the baselines. For $\AQ$, this is because the exponential mechanism must be run on the entire dataset $m$ times. For $\SliceEM$, the main cost comes from the $2m$ comparisons per record required to assign values to slices. In contrast, \name~runs the exponential mechanism only on small slices and requires just $\log m$ comparisons per record to determine slice membership, yielding substantial gains. 

Except for $\AQ$, the performance of all protocols improves slightly with increasing $\varepsilon$, since larger $\varepsilon$ reduces slice sizes (right side of Fig.~\ref{fig:time}). We provide the corresponding plots of total communication in App.~\ref{app-plots}.

Fig. \ref{fig:breakdown} provides the break down of the costs of the different stages of \name. Since the comparison operations necessary for bucketing scale with $\log m$, while the comparisons for $\SliceEM$ scale linearly with $m$, the proportion of the time spent on those comparisons increases from 42\% on average to 67\% on average as $m$ grows from 5 to 20. 

\noindent\textbf{Interactive \name.}  Allowing one additional round of interaction in \name~provides a performance gain. In the non-interactive setting, bucketing requires $\Omega(n'\log m)$ secure comparisons, where $n'$ is the total number of real and dummy records. With an extra round of communication, clients can directly supply bucket assignments, and the servers need only $2n'$ comparisons to verify correctness. This makes the bucketing cost independent of $m$ (Fig. \ref{fig:breakdown}), yielding up to a $1.33\times$ efficiency improvement (Fig.~\ref{fig:time}). 

\begin{figure*}
	\centering
        \includegraphics[width=0.37\linewidth]{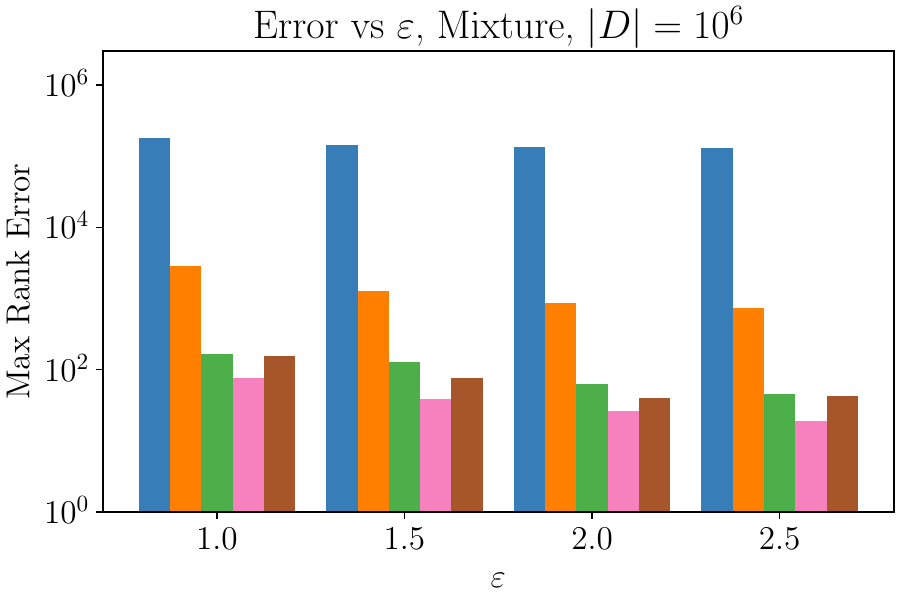}
        \includegraphics[width=0.37\linewidth]{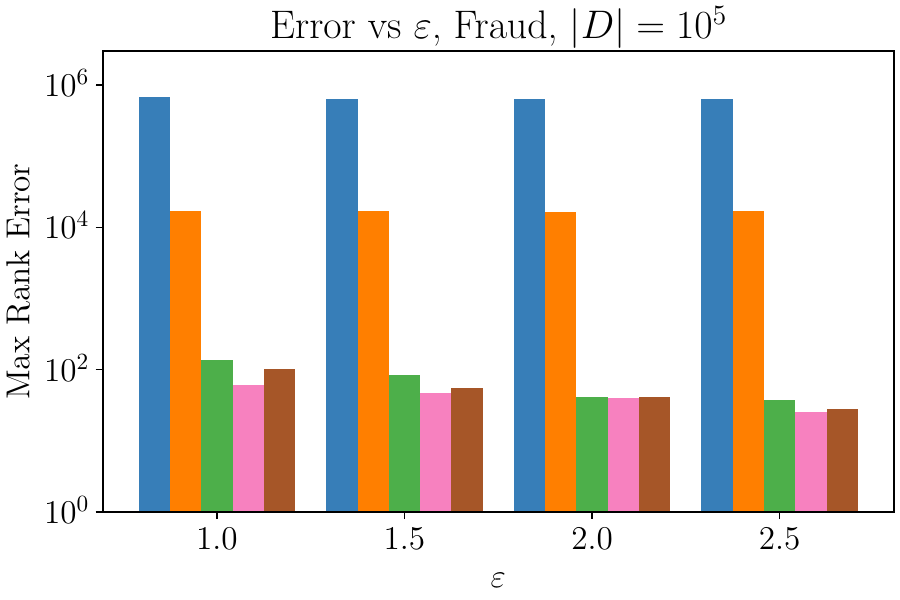}
        \includegraphics[width=0.24\linewidth]{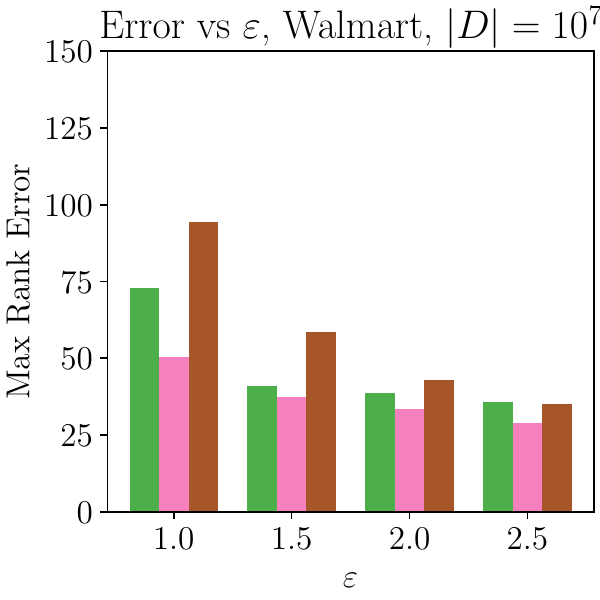}
    \includegraphics[width=0.37\linewidth]{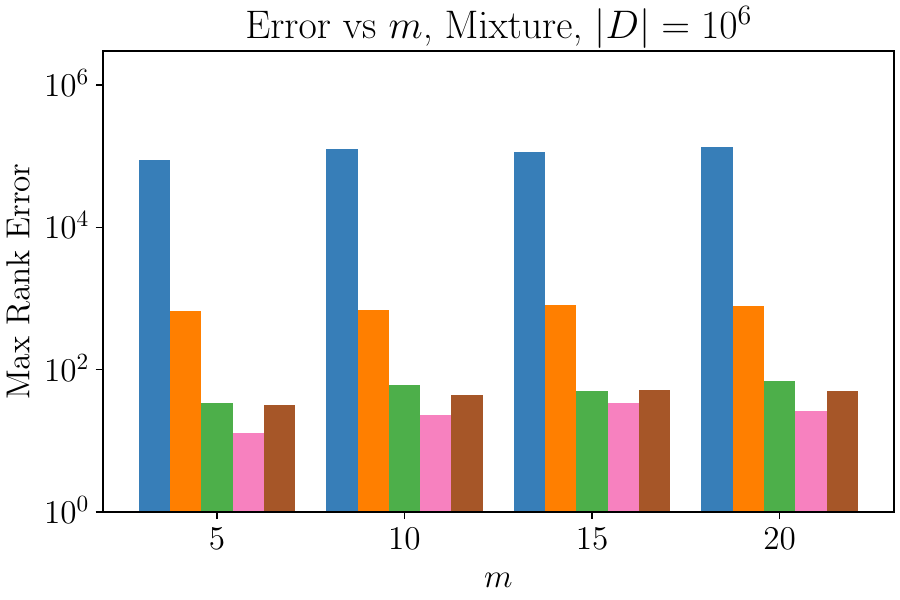}
        \includegraphics[width=0.24\linewidth]{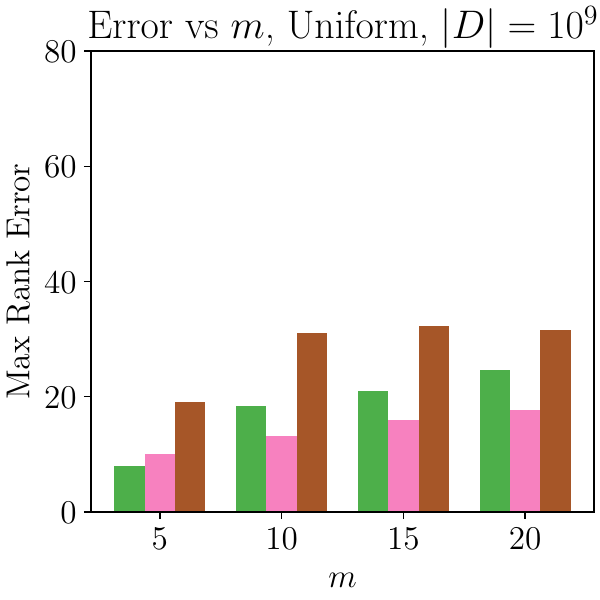}
        \includegraphics[width=0.24\linewidth]{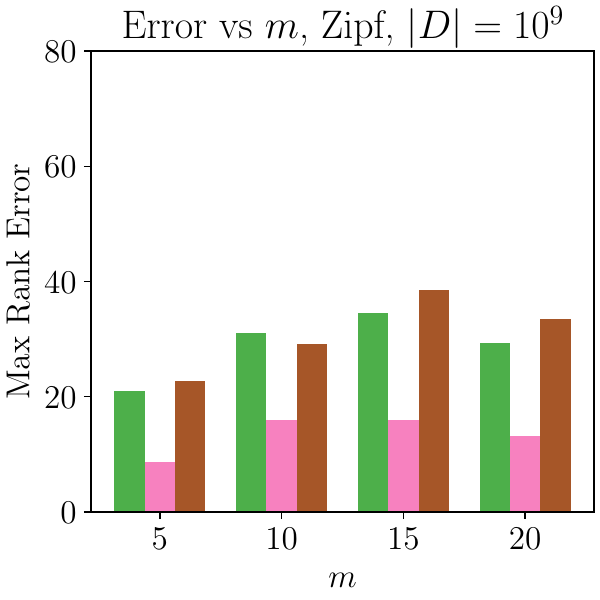}
        \includegraphics[width=0.6\linewidth]{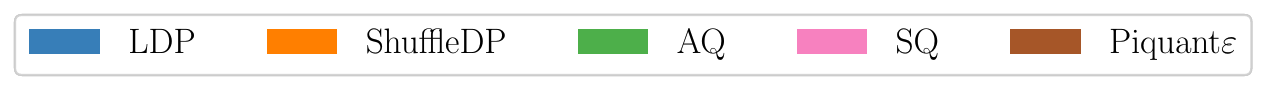}
\caption{Utility plots of quantile estimation methods against $\epsilon$ (first three plots) and $m$ (last three plots) for dataset size $n=10^6$. The domain $D$ depends on the dataset. For datasets with larger $D$, $\LDP$ and $\Shuffle$ are not plotted due to intractable running times.}
  \label{fig:uti-eps}
\end{figure*}
\subsection{Utility Evaluation} 
\noindent\textbf{Baseline.} We compare the error due to privacy of $\name$
against the aforementioned $\SliceEM$ and $\AQ$ algorithms. 
Additionally, we compared against an LDP and shuffle DP baseline (abbreviated as  $\LDP$ and $\Shuffle$, respectively). In both cases, we used the Hierarchical Mechanism~\cite{cormode2021frequency}, a c.d.f. estimation algorithm
based on decomposing the domain $\B$ into a tree. Note that the shuffle model requires a trusted shuffler and as such requires
more trust than the two-server model. Nevertheless we include the comparison for completeness.
The server running time for the Hierarchical Mechanism scales with $O(|\B| \log |\B|)$---this proved to be practically intractable for domain sizes larger than of $10^{6}$.

\noindent\textbf{Results.} Fig.\ref{fig:uti-eps} shows the quantile error versus $\epsilon$ for $\name$ and the four baselines on the Mixture, Fraud, and Walmart datasets (additional datasets in App. \ref{app-plots}). Errors are averaged over 10 runs. First, $\LDP$ and $\Shuffle$ incur errors orders of magnitude higher than the other methods. The error of $\LDP$ is consistently $10^4$-fold worse, corresponding to a normalized error over $10\%$, and reaches nearly $80\%$ for the skewed Fraud dataset. The $\Shuffle$ baseline has $18.5\times$ and $372\times$ higher error on the Mixture and Fraud datasets compared to $\name$. This is because the error of these methods scales with $\log(|\B|)^2$.

Next, we compare \name~with the central DP algorithms.  $\SliceEM$ generally achieves the lowest rank error. Compared to $\SliceEM$, $\AQ$ has up to $2.6\times$ higher error, consistent with~\cite{imola2025differentiallyprivatequantilessmaller} and supporting the choice of $\SliceEM$ as the inspiration for \name. While asymptotically $\name$ matches the accuracy of $\SliceEM$ (Thm.~\ref{thm:main-util}), in practice its error is slightly higher due  to the addition of two copies of noise for the malicious model and privacy budget splitting, though the increase remains modest. Specifically, \name~incurs an average of $1.94, 1.23, 1.48\times$ higher rank error on the Mixture, Fraud, and Walmart datasets, respectively. Across all datasets, the utility penalty of \name~relative to $\SliceEM$ is $2.41, 1.73, 1.43, 1.48\times$ for $\epsilon = 1.0, 1.5, 2.0, 2.5$, showing that the extra noise has a greater impact at smaller $\epsilon$. Overall, all three methods achieve modest rank errors; for $\epsilon = 1$, the average rank errors across all datasets correspond to a normalized error below $0.011\%$.

The last three plots in Fig.~\ref{fig:uti-eps} show the effect of $m$ on error. The plots demonstrate that the relative performance of the algorithms is roughly preserved for all tested values of $m$: $\SliceEM$ tends to perform the best, while $\AQ$ performed $2.33\times$ worse on average, $\name$ performed $2.61\times$ worse on average, and $\Shuffle$ and $\LDP$ performed several orders of magnitude worse. The algorithm errors generally increase sublinearly with more quantiles, though because the queried quantiles are different for each value of $m$, the pattern is somewhat noisy. 

\section{Related Work}

\noindent\textbf{Private quantile estimation.} 
State-of-the-art LDP quantile estimation uses the hierarchical mechanism~\cite{cormode2021frequency}, which estimates the full cumulative distribution function but incurs $O\big(\frac{\log^2 |D|}{\epsilon} \sqrt{n}\big)$ error. In the shuffle model, this can be improved to $O(\frac{\log^2 |D|}{\epsilon})$\cite{shuffle-hiding}. For a single quantile,\cite{aamand2025lightweight} removes a $\log|D|$ factor, but for multiple quantiles, error grows linearly with $m$. In the central model, prior works~\cite{tzamos2020optimal, alabi2022bounded, asi2020near, gillenwater2021differentially},  require distributional assumptions or lack formal guarantees. The closest work to ours is~\cite{kaplan2022differentially} and~\cite{imola2025differentiallyprivatequantilessmaller}; the former uses recursive estimation with error $O(\log|D| \log^2(m)/\epsilon)$, while the latter achieves $O(\log |D|/\epsilon + \log^2(m)/\epsilon)$, forming the basis for our $\FSEM$ design.

\noindent\textbf{Distributed DP}. 
Recent work has explored intermediate trust models for DP between the local and central settings~\cite{bittau_prochlo_2017,erlingsson_amplification_2019,cheu_distributed_2019,steinke_multi-central_2020,CryptE}. The shuffle model~\cite{bittau_prochlo_2017,cheu_distributed_2019} assumes a trusted shuffler, while distributed DP~\cite{EC:DKMMN06,steinke_multi-central_2020,cheu_et_al:LIPIcs.ITCS.2023.36} uses MPC to share computation across servers. A common example is secure aggregation, where only sums are computed~\cite{goryczka_comprehensive_2017,CCS:BIKMMP17,mugunthan2019smpai,apple_exposure_2021}.

Prior work on more complex functions typically samples noise inside MPC~\cite{EC:DKMMN06,Alhadidi,Eigner,Bhler2020SecureST,boehlher}, a communication-heavy technique we avoid via local noise sampling. To our knowledge, no prior work addresses distributed computation of multiple DP quantiles. The closest is~\cite{BohlerK20}, which supports only a single quantile, assumes a semi-honest model, and requires client interaction.



\section{Conclusion}
We have demonstrated that maliciously secure protocols for differentially private quantile estimation can scale to one million clients with practical running times. Our approach builds on advances in central-model quantile estimation while addressing MPC-specific challenges: using sampling to cut costs without harming accuracy, private splitting to reduce multi-quantile estimation to single-quantile tasks, a partial sorting technique to minimize comparisons, and optimizations that avoid secure exponentiation in the exponential mechanism.
\section*{Acknowledgments}

Jacob Imola and Rasmus Pagh were supported by a Data Science Distinguished Investigator grant
from Novo Nordisk Fonden, and are part of BARC, supported by the VILLUM Foundation grant 54451

Hannah Keller was supported by the Danish Independent Research Council under Grant-ID DFF-2064-00016B (YOSO).

Fabrizio Boninsegna was supported in part by the MUR PRIN 20174LF3T8 AHeAD project, and by MUR PNRR CN00000013 National Center for HPC, Big Data and Quantum Computing.



\bibliographystyle{plain}
\bibliography{cite, crypto}

\appendix 
\section{Background Cntd.}\label{app:back}
\begin{theorem}
    [Parallel Composition]
Let $\X\in \B^n$ be a dataset. For each $i \in [k]$, let $\X_i$
subset of $\X$, let $\mathcal{M}_i$ be a mechanism that takes $\X\cap \X_i$
as input, and suppose $\mathcal{M}_i$ is $\epsilon_i$ differentially private. If
$ \X_i \cap \X_j= \varnothing$ whenever $i \neq j$, then the composition of the
sequence $\mathcal{M}_1, \cdots, \mathcal{M}_k$ is max $\{\epsilon_i\}$-DP.

    \label{thm:parrallel}
\end{theorem}

When applied multiple times, the DP guarantee degrades gracefully as follows.

\begin{theorem}[Sequential Composition] \label{theorem:seq}
If $\mathcal{M}_1$ and
$\mathcal{M}_2$ are $\epsilon_1$-DP and $\epsilon_2$-DP algorithms that use independent randomness, then releasing the outputs $(\mathcal{M}_1(\X),\mathcal{M}_2(\X))$ on database $\X$ satisfies $(\epsilon_1+\epsilon_2)$-DP.\label{thm:sequence}\end{theorem} 
Note that there are more advanced versions of composition as well~\cite{dwork2010differential}.

\section{Omitted Functionalities and Protocols}\label{app:func}

\begin{mybox2}{Functionality: $\Fshuffle$}
\textbf{Input:} Secret shared dataset $\langle \X \rangle_b$ from $\Ser_b$\\
\begin{enumerate}[nosep, label=\arabic*.]
\item Generate a random array $R \leftarrow \F^{|\X|}$
\item Re-randomize the input shares as $\langle \X \rangle + \langle t \rangle$
\item Select a random permutation from the symmetric group of size $\pi \in S_||$
\item Return $\langle \pi(X)\rangle_b$ to $\Ser_b$
\end{enumerate}
\end{mybox2}
\noindent\begin{minipage}{\columnwidth}
\captionof{figure}{Ideal Functionality  $\Fshuffle$}\label{fig:IF:shuffle}
\end{minipage}
\newpage

\begin{mybox}{Protocol: \textsf{SecureLinearSearch} $\Pi_{\textsf{LinearSearch}}$}
\textbf{Input:} Secret shared list $\langle L \rangle $; Secret shared value $\langle s \rangle $;\\
\textbf{Output:} Secret shared index $\langle ind\rangle$ such that $L[ind] \leq  s \leq L[ind+1] $
\\\textbf{Initialization:}
\begin{enumerate}[nosep, label=\arabic*.] 
\item $ind=0$
\item \textsf{Share}($ind$)
\end{enumerate}
Each server $\Ser_b$
\begin{enumerate}[nosep, label=\arabic*.] 
 \setItemnumber{3}
\item \textbf{For} $i \in [|L|]$
\item \hspace{0.5cm} $\langle t_1 \rangle = \Cmp(\langle L[i]\rangle, \langle s \rangle)$
\item \hspace{0.5cm} $\langle t_2 \rangle = \Cmp(\langle s \rangle, \langle L[i+1] \rangle)$
\item \hspace{0.5cm} $\langle t_3 \rangle = \Mult(  \langle t_1 \rangle,  \langle t_2 \rangle)$
\item \hspace{0.5cm} $\langle ind \rangle = \langle ind \rangle + \Mult( \langle t_3 \rangle, i)$
\item \textbf{EndFor}

\end{enumerate}
\textbf{Output} $\langle ind \rangle$
\end{mybox}
\noindent\begin{minipage}{\columnwidth}
\captionof{figure}{Protocol $\Pi_{\textsf{LinearSearch}}$}\label{fig:prot:LinearSearch}
\end{minipage}

\begin{mybox}{Protocol: Secure Exponential Mechanism $\Pi_{\EM}$}
\textbf{Input.} Secret-shared dataset in sorted order $\langle \Xs \rangle$ where $X \in \B^n $ with $\B=[0,\d]$; Queried quantile $q \in (0,1)$; Privacy parameter $\epsilon$; 
\\\textbf{Output.} Noisy quantile estimate $\a$
\begin{enumerate}[nosep, label=\arabic*.]
\item $\rq=\lfloor qn \rfloor$\\
\textcolor{blue}{ \hspace{3cm}$\rhd$} Weight computation
\item $\langle T_1[1]\rangle = \langle \Xs[1] \rangle$
\item $\langle T_1[n+1]\rangle = (\d-\langle \Xs[n] \rangle)$ 
\item $\langle T_2[1]\rangle$ = $\langle 0 \rangle$
\item $\langle T_2[n+1]\rangle = \langle \Xs[n] \rangle $
\item $\langle wt[1]\rangle = \exp(-\frac{\epsilon}{2}|1- r_q|)\cdot\langle T_1[1] \rangle$
\item $\langle wt[n+1]\rangle = \exp(-\epsilon |n+1-\rq|)\cdot(\langle T_{n+1}\rangle)$
\item \textbf{For } $i \in [1, n-1]$
\item \hspace{0.2cm} $\langle T_1[i+1]\rangle = (\langle \Xs[i+1] \rangle -\langle \Xs[i]\rangle)$
\item \hspace{0.2cm} $\langle T_2[i+1]\rangle = \langle \Xs[i]\rangle$
\item \hspace{0.2cm} $\langle wt[i+1] \rangle = \exp(-\frac{\epsilon}{2}|i-\rq|)\cdot \langle T_1[i+1]\rangle $
\item \textbf{End For}
\item $\langle z \rangle = \sum_{i=1}^{n+1}\langle  p[i] \rangle$
\item $\langle t\rangle = \Frand(l)$
\item $\langle u \rangle = \Mult(\langle z\rangle, \langle t\rangle)$
\item $\langle ind \rangle$ = $\Pi_\textsf{LinearSearch}(\langle wt\rangle,\langle u\rangle )$
\item $\langle \a \rangle= \Pi_\textsf{IntervalSearch}(\langle T_1\rangle, \langle T_2\rangle, \langle ind \rangle)$
\item \textbf{Return} $\Rec(\langle \a \rangle)$
\end{enumerate}

\end{mybox}
 \noindent\begin{minipage}{\columnwidth}
\captionof{figure}{Protocol  $\Pi_{\EM}$}\label{fig:prot:EM}
\end{minipage}

\newpage

\begin{mybox}{Protocol: Secure Sampling $\Pi_{\textsf{Sample}}$}
\textbf{Parameters:} $k$ - Number of i.i.d samples to be generated; $l$ - Bit length
\\
\textbf{Input:} Secret-shared dataset $\langle \X \rangle$\\
Each server $\Ser_b$
\begin{enumerate}[nosep, label=\arabic*.]
\item $t = 0$
\item $G=\varnothing$
\item \textbf{While} $t < k$
\item \hspace{0.5cm} $\langle r \rangle = \Frand(l)$
\item \hspace{0.5cm} $s=\lfloor \Rec(\langle r \rangle ) \cdot n\rfloor$
\item \hspace{0.5cm}\textbf{If} $s \not \in G$
\item \hspace{0.7cm} $G=G\cup s$
\item \hspace{0.7cm} $t=t+1$
\item \hspace{0.5cm}\textbf{EndIf}
\item \textbf{EndWhile}
\end{enumerate}
\textbf{Output} $\bigcup_{s \in G}\langle \X_{s}\rangle$
\end{mybox}
\noindent\begin{minipage}{0.5\textwidth}
\captionof{figure}{Protocol $\Pi_{\textsf{Sample}}$}\label{fig:prot:sampling}
\end{minipage}

\begin{mybox}{Protocol: Secure Dummy Record Check $\Pi_{\text{Dummy}}$}
\textbf{Input:}  List of secret shared dummy records $\langle R \rangle$;  Values of dummy records $A$; Number of bins $k$; Parameter $c$\\
\textbf{Output:} 1/0 \\
\textbf{Initialize:}
\begin{enumerate}[nosep, label=\arabic*.]
\item $\Cnt=[0, 
\cdots, 0]$ \\\hfill \textcolor{blue}{$\rhd$} Initialize an array of zeroes of size $k$ 
\item $\textsf{Sum}=0$
\item $\Flag = [0, \cdots, 0]$ \\\hfill \textcolor{blue}{$\rhd$} Initialize an array of zeroes  of size $|R|$ 
\item \textsf{Share}( $\Cnt$)
\item \textsf{Share}($\Flag$)
\item \textsf{Share}($\textsf{Sum}$)
\end{enumerate}
Each server $\Ser_b$:
\begin{enumerate}[nosep, label=\arabic*.]
 \setItemnumber{5}
\item \textbf{For} $i \in [|R|]$
\item \hspace{0.5cm}\textbf{For} $ j \in [k]$
\item \hspace{1cm} $\langle t \rangle = \Equal(\langle R[i] \rangle, A[j])$
\item \hspace{1cm} $\langle \Flag[i] \rangle = \langle \Flag[i] \rangle + \langle t \rangle$
\item \hspace{1cm} $\langle \Cnt[j] \rangle = \langle \Cnt[j] \rangle +\langle t\rangle$
\item \hspace{0.5cm} \textbf{EndFor}
\item \hspace{0.5cm} \textbf{If} $(\Rec(\langle \Flag[i]  \rangle)==0)$ \\\hfill \textcolor{blue}{$\rhd$} Checking the value of the dummy record 
\item \hspace{1cm} \textbf{Return} 0
\item \hspace{0.5cm} \textbf{EndIf}
\item \textbf{EndFor} 
\item \textbf{For} $i \in [1,k]$
\item \hspace{0.5cm} $\langle \textsf{Sum} \rangle =\langle \textsf{Sum} \rangle +\langle\textsf{Count}[i]\rangle$
\item \hspace{0.5cm} \textbf{If} $\Big(!(\Rec(\Cmp(\Abs(\langle \textsf{Sum}\rangle - 2i\cdot c, c)) \text{ OR } \Rec(\Cmp(\Cnt [i], 4c)))\Big )$
\item \hspace{1cm} \textbf{Return} 0
\item \hspace{0.5cm} \textbf{EndIf}
\item \textbf{EndFor}
\item \textbf{Return} 1
\end{enumerate}
\end{mybox}
\noindent\begin{minipage}{\columnwidth}
\captionof{figure}{Protocol $\Pi_{\text{Dummy}}$}\label{fig:prot:dummy}
\end{minipage}

\begin{mybox}{Protocol: \textsf{SecureIntervalSearch} $\Pi_{\textsf{IntervalSearch}}$}
\textbf{Parameters:}  Bit length $l$
\\
\textbf{Input:} Secret shared list $T_1$ of interval lengths; Secret shared list $T_2$ of the lowest value of the intervals; Secret shared index $\langle  ind \rangle$ ;\\
\textbf{Output:}
\\\textbf{Initialization:}
\begin{enumerate}[nosep, label=\arabic*.] 
\item Set $val_1=0, val_2=0$
\item \textsf{Share}($val_1$)
\item \textsf{Share}($val_2$)
\end{enumerate}
\textbf{Protocol:}\\
Each server $\Ser_b$
\begin{enumerate}[nosep, label=\arabic*.] 
 \setItemnumber{4}
\item $\langle r \rangle = \Frand(l)$
\item \textbf{For} $i \in [|L_1|]$
\item \hspace{0.5cm} $\langle t \rangle = \Equal(\langle ind \rangle, i)$
\item \hspace{0.5cm} $\langle val_1 \rangle = \langle val_1 \rangle + \Mult( \langle t \rangle, \langle L_1[i] \rangle)$
\item \hspace{0.5cm} $\langle val_2 \rangle = \langle val_2 \rangle + \Mult( \langle t \rangle, \langle L_2[i] \rangle)$
\item \textbf{EndFor}
\item $\langle ans \rangle = \langle val_1 \rangle  + \Trunc(\FLMult(\langle val_2 \rangle, \langle r \rangle, l )$

\end{enumerate}
\textbf{Output} $\langle ans \rangle$
\end{mybox}
\noindent\begin{minipage}{\columnwidth}
\captionof{figure}{Protocol $\Pi_{IntervalSearch}$}\label{fig:prot:interval}
\end{minipage}

\subsection{Authenticated Secret Shares}\label{app:ass}
Here we describe the \textsf{Share}(x) protocol with which a client can send authenticated secret shares of its private values to the servers. 
The third-party dealer selects the global MAC $\alpha$ and a random value $r$. Let $\alpha_0+\alpha_1=\alpha$ and $r_0+r_1=r$. The dealer distributes $z_b=\alpha_b\cdot r+\delta_b$ to server $\Ser_b$ s.t. $\delta_0+\delta_1=0$. To obtain the shares of $x$ from a client, server  $\Ser_b$ sends $r_b$ to the client. The client then reveals $y=x-(r_1+r_0)$. Note that in order to prevent the dealer from seeing this, the client may encrypt $y$ using shared secret keys established with each server and send the ciphertexts to them.
Server $\Ser_0$ sets the shares as $[x]_0=r_0, [\alpha x]_0=\alpha_0.y+z_0$ while $\Ser_0$ sets the shares as $[x]_0=y+r_1, [\alpha x]_1=\alpha_1.y+z_1$. It's easy to see that the shares are now in the correct form.

\section{Omitted Details from Continual Counting}\label{app:cont-count}
\def \CCd{\mathsf{CD}}
In this section, we show how we run the continual counting mechanisms  $\CC$ and $\CCd$. 

The first continual counting mechanism, which we denote by $\CC_\epsilon$, is the the binary tree mechanism~\cite{dwork2010differential, chan2011private}, which works by constructing a segment tree whose leaves are the intervals $[0,1), \ldots, [m-1, m)$, and sampling a Laplace random variable $\nu_I$ for each node $I$ of the tree. For $T=\lceil\log_2(m+1)\rceil$, let $I_1,\ldots, I_{T}$ denote the interval decomposition of $[0,i)$. 
Then, the total noise $\eta_i$ at index $i$ is given by $\sum_{j = 1}^T \eta_{I_j}$. In~\cite{imola2025differentiallyprivatequantilessmaller}, it is shown that rounding the binary tree mechanism makes it integer-valued, while preserving its privacy and utility properties. We are able to easily adapt that result to our setting. In the following, let $\mathsf{Cont}(m) \subseteq \{0,1\}^m$ denote the set of vectors $v$ such that there exist indices $1 \leq s \leq t \leq n$ such that $v_i = \mathbf{1}[s \leq i \leq t]$.

\begin{lemma}\label{lem:co} 
    There exists a noise distribution $\CC_\varepsilon(m)$ (written $\CC$ for short) supported on $\mathbb{Z}^m$ such that 
    \begin{itemize}
        \item For all vectors $\boldsymbol{v} \in \Z^m$ and $\mathbf{e} \in \mathsf{Cont}(m)$, we have 
    \begin{gather*}
        \Pr_{\boldsymbol{\eta} \sim \CC}[\boldsymbol{\eta}=\boldsymbol{v}] \leq e^{\epsilon} \Pr_{\boldsymbol{\eta} \sim \CC} [ \boldsymbol{e} + \boldsymbol{\eta} = \boldsymbol{v}] \\
        \Pr_{\boldsymbol{\eta} \sim \CC}[\boldsymbol{e}+\boldsymbol{\eta}=\boldsymbol{v}
        ] \leq e^{\epsilon} \Pr_{\boldsymbol{\eta} \sim \CC} [ \boldsymbol{\eta} = \boldsymbol{v} ]
   \end{gather*}
    \item For all $\beta > 0$, $\Pr_{\boldsymbol{\eta} \sim \CC}[\|\boldsymbol{\eta}\|_\infty \geq \frac{6}{\epsilon}\log(m) \log(\frac{2m}{\beta})] \leq \beta$.
    \end{itemize}
\end{lemma}
\begin{proof}
    This is immediate from~\cite[Lemma 3.1]{imola2025differentiallyprivatequantilessmaller}. That result applies to add/remove privacy; since we are using substitute privacy (with fixed $n$), we need to apply the lemma twice, with budget $\frac{\epsilon}{2}$, and then use group privacy.
\end{proof}

Our second continual counting algorithm, which we denote by $\CCd$ for \emph{continual differences}, will be generated instead by first generating $\boldsymbol{\eta} \sim \CC_\epsilon(m)$, truncating each $\eta_i$ to $[-6\log(m) \log(\frac{2m}{\delta})/\epsilon, 6\log(4m) \log(\frac{2m}{\delta})/\epsilon]$, and then by defining $\nu_i =12 \log(m) \log(\frac{2m}{\delta})/\epsilon +  \eta_i-\eta_{i-1}$ (taking $\eta_0 = 0$). This simple transformation satisfies the following property:

\begin{lemma}\label{lem:cd}
    There exists a noise distribution $\CCd_\varepsilon(m)$ (written $\CCd$ for short) supported on $\Z^m$ such that, for all $\delta, \beta > 0$,
    \begin{itemize}
    \item For subsets $V \subseteq \Z^m$ and $\mathbf{e} \in \{-1,0,1\}^n$ such that there exist two indices $1 \leq s \leq t \leq n$ such that $e_s = -1$, $e_t = 1$, and $e_i = 0$ otherwise, we have 
    \begin{gather*}
        \Pr_{\boldsymbol{\nu} \sim \CCd}[\boldsymbol{\nu}\in V] \leq e^{\epsilon} \Pr_{\boldsymbol{\nu} \sim \CCd} [\boldsymbol{e} + \boldsymbol{\nu} \in V]+\delta \\
        \Pr_{\boldsymbol{\nu} \sim \CCd}[\boldsymbol{e}+\boldsymbol{\nu}\in V] \leq e^{\epsilon} \Pr_{\boldsymbol{\nu} \sim \CCd} [ \boldsymbol{\nu} \in V ] + \delta
    \end{gather*}
    \item For all indices $1 \leq i \leq m$, with probability at least $1-\beta$, \begin{gather*}
    \left\lvert\textstyle{\sum_{j=1}^i} \nu_j - \frac{12i}{\epsilon}\log (m)\log(\frac{2m}{\delta})\right\rvert \leq \tfrac{6}{\epsilon} \log(m) \log(\tfrac{2m}{\beta}), \\
    \text{ and }\nu_i \in [0, \tfrac{24}{\epsilon} \log(m) \log(\tfrac{2m}{\delta})].
    \end{gather*}
    \end{itemize}
\end{lemma}
\begin{proof}

    To show the first property, suppose there is \emph{no} truncation step. By Lemma~\ref{lem:co}, a pure DP holds guarantee of the form 
    
        \begin{gather*}
        \Pr_{\boldsymbol{\nu} \sim \CCd}[\boldsymbol{\nu} = \boldsymbol{v}] \leq e^{\epsilon} \Pr_{\boldsymbol{\nu} \sim \CCd} [\boldsymbol{e} + \boldsymbol{\nu} = \boldsymbol{v}] \\
        \Pr_{\boldsymbol{\nu} \sim \CCd}[\boldsymbol{e}+\boldsymbol{\nu}= \boldsymbol{v}] \leq e^{\epsilon} \Pr_{\boldsymbol{\nu} \sim \CCd} [ \boldsymbol{\nu} = \boldsymbol{v} ].
    \end{gather*}
    By the utility guarantee of Lemma~\ref{lem:co}, truncation only shifts the at most $\delta$ mass in the distributions of $\boldsymbol{\nu}, \boldsymbol{e} + \boldsymbol{\nu}$. Thus $(\epsilon, \delta)$-DP holds.

    The second property is immediate from applying the utility guarantee of Lemma~\ref{lem:co} with the parameter $\beta$.
\end{proof}
\section{Additional Plots}\label{app-plots}

\begin{figure}
	\centering
            \includegraphics[width=0.8\linewidth]{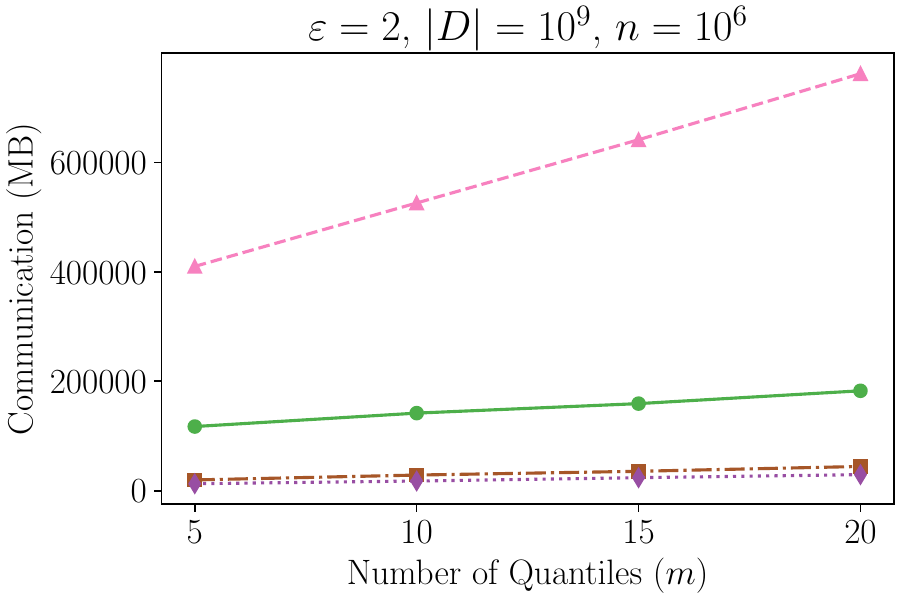}
            \includegraphics[width=0.8\linewidth]{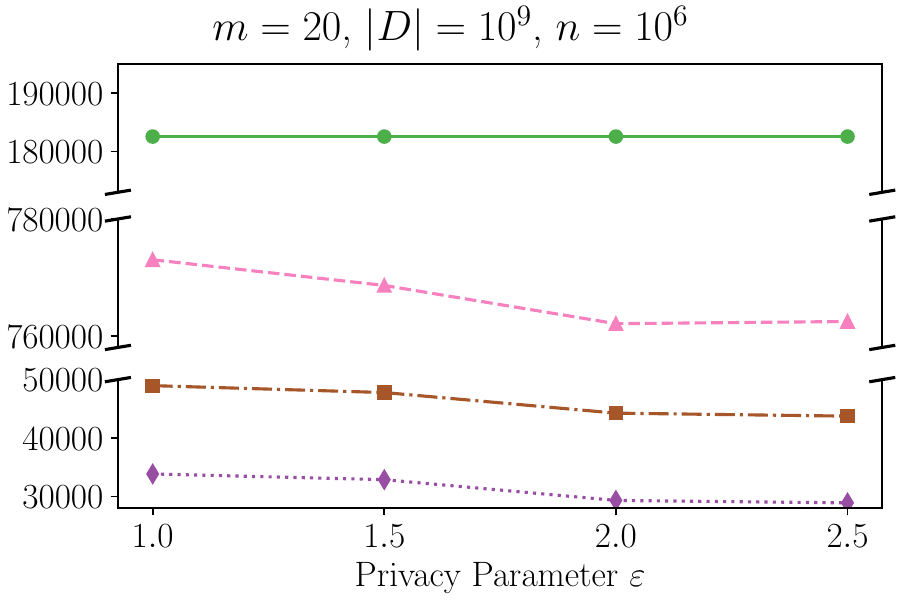}
            \includegraphics[width=0.6\linewidth]{figures/CCS/legend_algos.pdf}%
    \caption{Communication cost benchmarks.}
    \label{fig:communication}
\end{figure}
\name~also offers substantial improvement in communication overhead over the baselines - up to over $\sim 21 \times$ and $\sim 6\times$ over SQ and AQ, respectively (Fig. \ref{fig:communication}). 

In Figure~\ref{fig:uti-eps-extra}, we show all utility data collected for the different algorithms. The same trends found in Section~\ref{sec:eval} hold. Some of the tested datasets had many repeated values, and we used an extended notion of rank error to evaluate the algorithms in the experiments. For example, if the $20\%$ to $25\%$ quantiles of a dataset are all the same value, and the algorithm returns this value, its empirical error is measured to be $0$. For this reason, on some datasets, including Mixture and Fraud datasets, particularly with a small number of quantiles, the error for some algorithms in Figure~\ref{fig:uti-eps-extra} evaluate to $0$.

Finally, when $m$ increases, the queried quantiles are not the same---for example, when $m=2$, the queried quantiles are $\{33\%, 67\%\}$, while for $m=3$, the quantiles are $\{25\%, 50\%, 75\%\}$, a completely disjoint set. As the datasets are non-uniform, different quantiles present different levels of difficulty for each algorithm. For this reason, the trend of quantile error increasing with $m$ is noisy, and there are spikes in the data such as on the Walmart dataset.

\begin{figure*}
	\centering
        \includegraphics[width=0.21\linewidth]{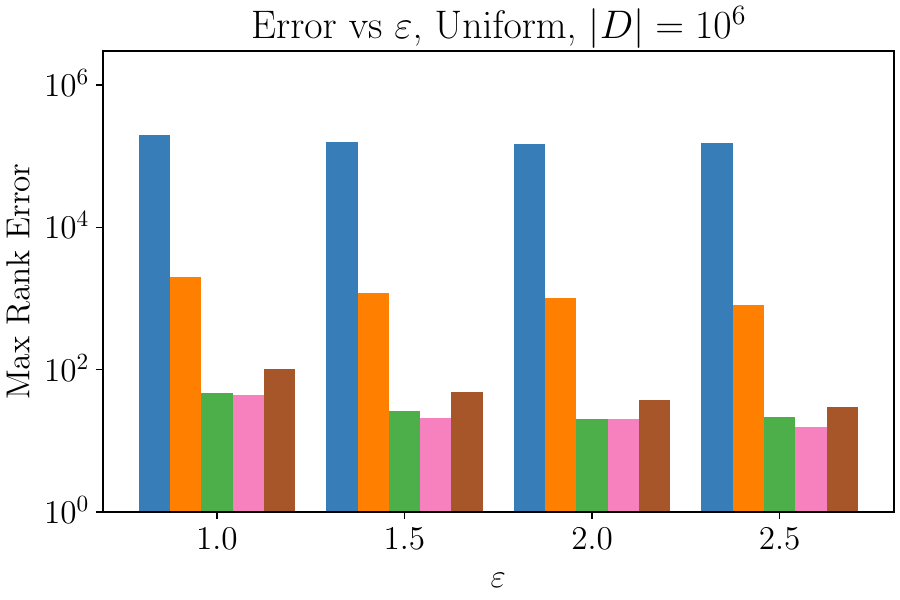}
        \includegraphics[width=0.21\linewidth]{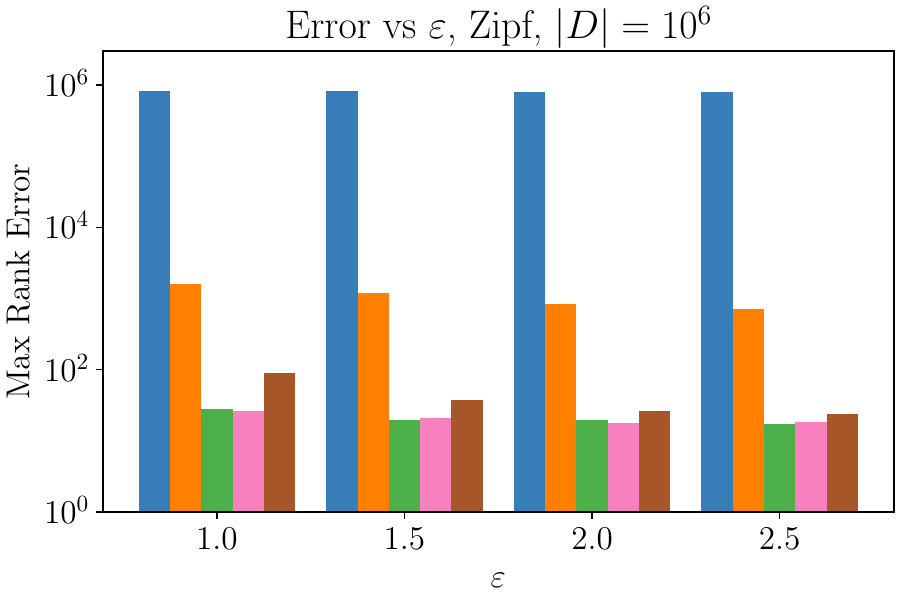}
        \includegraphics[width=0.16\linewidth]{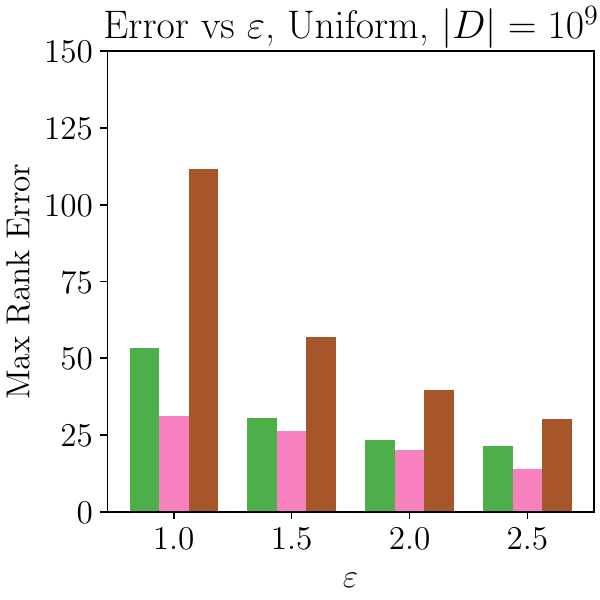}
        \includegraphics[width=0.16\linewidth]{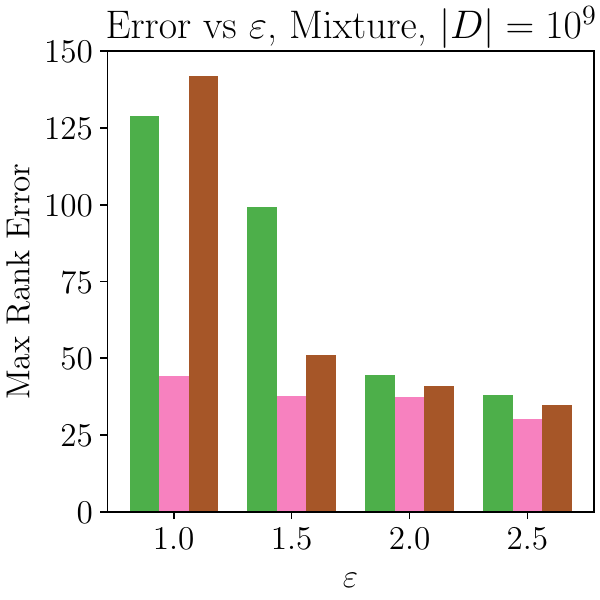}
        \includegraphics[width=0.16\linewidth]{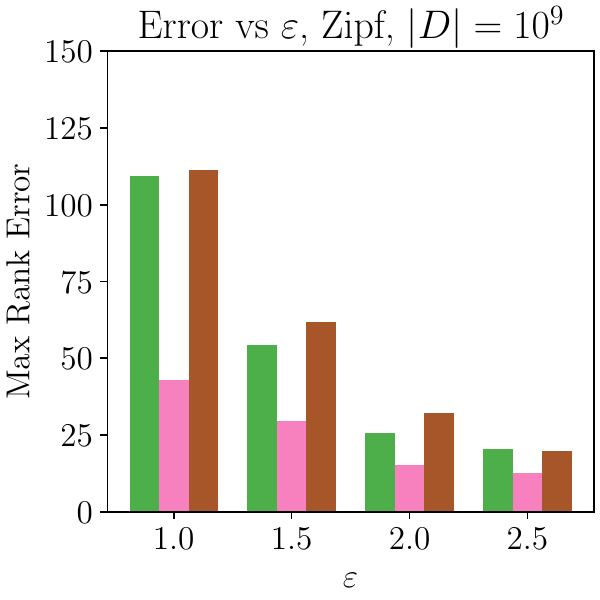}
        \includegraphics[width=0.21\linewidth]{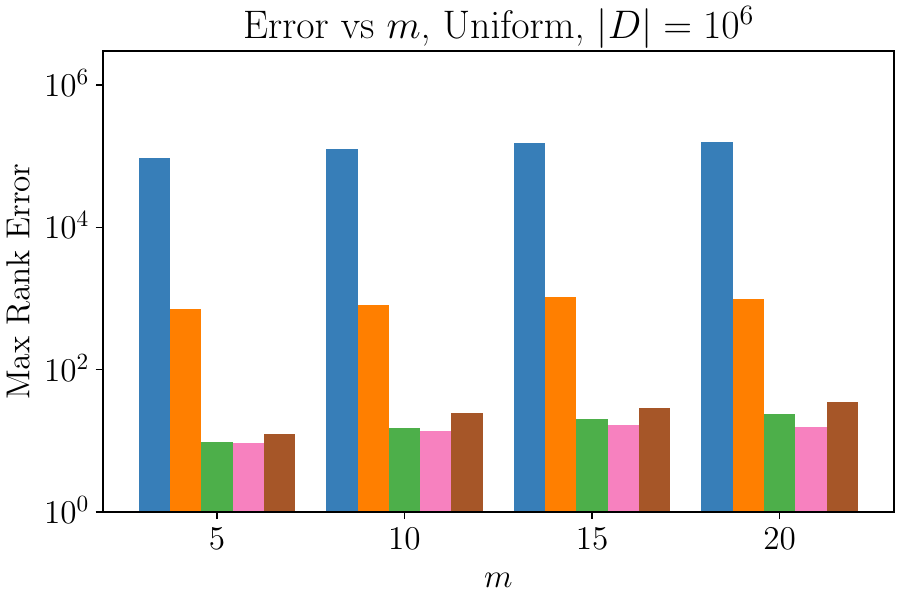}
        \includegraphics[width=0.21\linewidth]{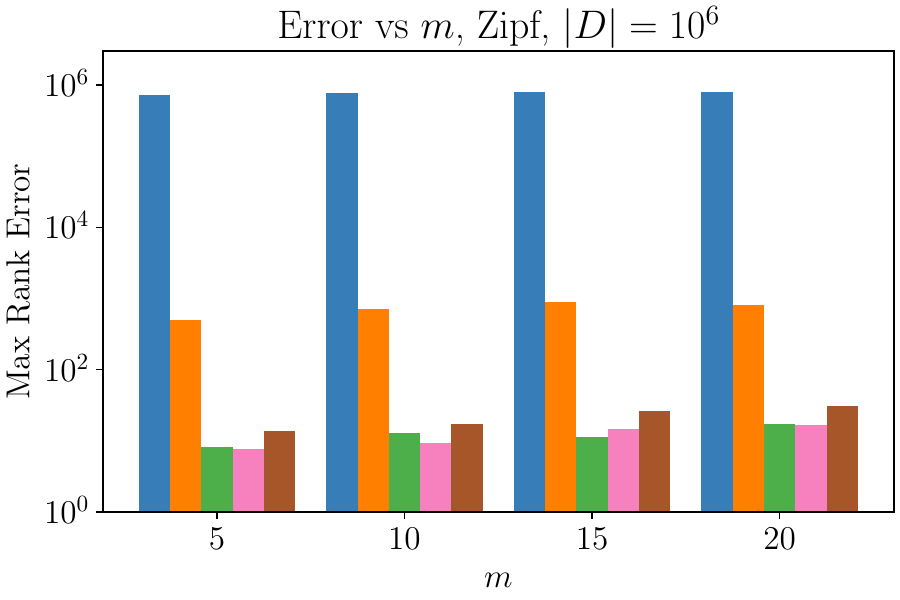}
        \includegraphics[width=0.21\linewidth]{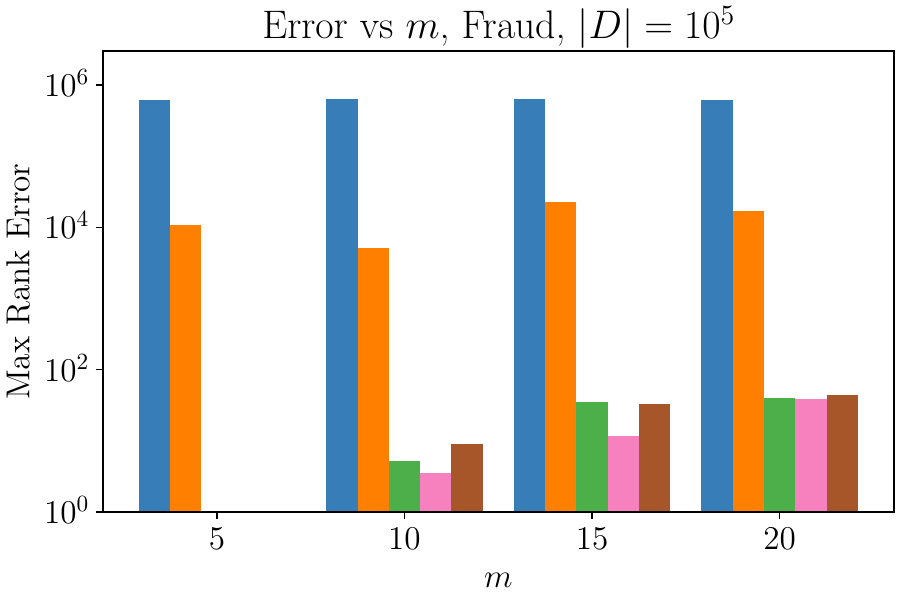}
        \includegraphics[width=0.16\linewidth]{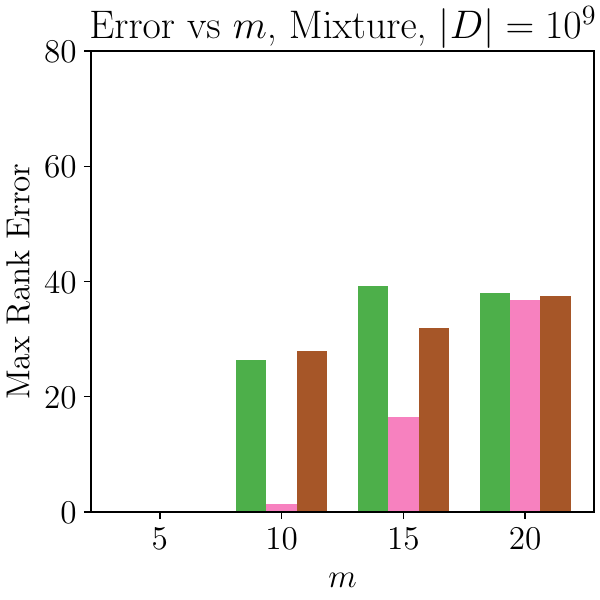}
        \includegraphics[width=0.16\linewidth]{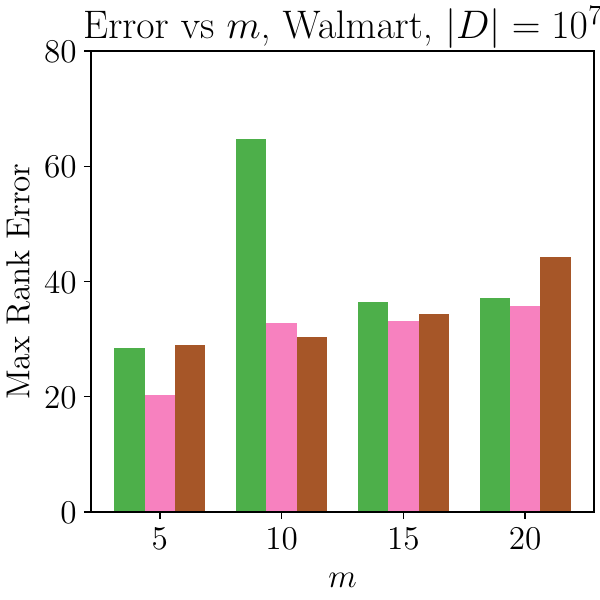}
        \includegraphics[width=0.5\linewidth]{figures/Utility/legend-small.pdf}
\caption{Additional utility plots of quantile estimation methods against $\epsilon$. The domain $D$ depends on the dataset. For datasets with larger $D$, LDP and Shuffle are not plotted due to infeasible running times.}
  \label{fig:uti-eps-extra}
\end{figure*}
\section{Guidelines for Setting Hyperparameters in ${\name}$}\label{sec:set-hyperp}

Here, we show how to fine-tune the parameters $\epsilon_1, \epsilon_2, \epsilon_3$, and $k$, given a total privacy budget $\epsilon$, as well as the parameters $n,m,\delta,\beta$ which are part of the input.

First, $\epsilon_1$ is subject to privacy amplification by subsampling~\cite[Theorem 10]{balle2018privacy}, and thus it can be amplified to the smaller value 
\[
f(\epsilon) = \ln\left(1 + \left(1-\left(1-\tfrac{1}{n}\right)^k\right)(e^{\epsilon_1}-1)\right).
\]

Because $\epsilon_1$ only affects the running time, and only once it is set sufficiently low, one can binary search for the value $\epsilon_1$ such that $\epsilon^* = f^{-1}(\epsilon_1)$ satisfies
\[
    \ln(\tfrac{1}{\beta})^{1/3} = \sqrt{\frac{\ln(|D|m/\beta) + \log(m)\log(m/\beta)}{\epsilon^*(nm)^{1/3}}}.
\]
According to Theorem~\ref{thm:runtime}, this is the minimum value of $\epsilon_1$ that will result in a running time that grows with $(nm)^{1/3} \ln(\frac{1}{\beta})^{1/3}$, and for most applications, due to the privacy amplification, it will use a very small portion of the privacy budget. Based on this way of choosing $\epsilon_1$, one should set $k$ roughly equal to $(mn)^{2/3}\ln(\frac{1}{\beta})$, and can fine-tune this value experimentally.

Having set $k$ and $\epsilon_1$, we found that the remaining privacy budget $\epsilon - \epsilon_1$ should be split roughly equally between $\epsilon_2, \epsilon_3$. In our experiments, we found that splits between $0.3$ and $0.7$ performed best.

\section{Privacy Proofs}
The first three proofs deal with the privacy of the three functionalities; the last deals with the security of $\Pi_{\name}$.
\subsection{Proof of Theorem \ref{thm:privacy:SEM}}
\begin{proof}
    We will show that for any quantiles satisfying the assumption gap of $\FSEM$, and any fixed values of $\eta_1^\calA, \ldots, \eta_m^{\calA}$ which do not cause $\FSEM$ to abort, $\FSEM$ has the same behavior as the $\mathsf{SliceQuantiles}$ algorithm run on the shifted quantiles $q_1 +\frac{w}{2n} - \frac{\eta_1^\calA}{n}, \ldots, q_m +\frac{w}{2n} - \frac{\eta_m^\calA}{n}$, except in total variation distance at most $\delta$. To show this, we note the following implementation differences in the two algorithms:

    \squishl
        \item The continual counting noise $\CC$ is truncated to be in the interval $[-\frac{w}{2}, \frac{w}{2}]$ and shifted up by $\frac{w}{2}$ in $\FSEM$, whereas it is not truncated in $\mathsf{SliceQuantiles}$, and the algorithm instead aborts (returns a random value) if the noise values are not observed to be in this range.
        \item $\mathsf{SliceQuantiles}$ requires the gap between quantiles to be $\frac{2}{n} (\frac{w}{2} + h + 1)$, whereas $\FSEM$ requires them to be at least $\frac{2}{n} (w + h + 1)$. 
    \squishe
    This justifies the quantile shifts of $\frac{w}{2n} - \frac{\eta^\calA}{n}$, as the $\frac{w}{2n}$ undoes the $\frac{w}{2}$ term added by $\FSEM$. Furthermore, $\mathsf{SliceQuantiles}$ will not reject the shifted quantiles as they satisfy its separation gap since each $\eta_i^\calA \in [0,w]$. Finally, the truncation step only affects $\mathsf{SliceQuantiles}$ in TVD $\frac{\delta}{2}$ by the choice of $w$.

    We can then directly apply   \cite[Theorem 4.2]{imola2025differentiallyprivatequantilessmaller} with $\epsilon_1 = \frac{\epsilon}{4}$ (since our implementation of $\CC$ satisfies substitute privacy already) and $\epsilon_2 = \frac \epsilon 6$ to establish that $\mathsf{SliceQuantiles}$ satisfies $(\epsilon, \frac{\delta}{2} e^\epsilon)$-DP. As $\FSEM$ only differs in TVD at most $\frac{\delta}{2}$, the resulting privacy guarantee is $(\epsilon, \delta e^\epsilon)$.
\end{proof}
\subsection{Proof of Theorem \ref{thm:privacy:bucketing}}
\begin{proof} 
Since the adversary has no information about $X$, the probability that the adversary makes the algorithm halt is the same whether $X,X'$ is used. Thus, it suffices to show privacy conditioned on the algorithm not aborting.

The view of the adversary is given by $\textsc{View}^\calA_{\FBuc}=(\sum_{i=1}^{2m
+1}\gamma_i, \boldsymbol{\gamma}^\calA, \mathbf{cnt})$. Let $X,X'$ denote two adjacent datasets, and define $b_i = \lvert B_i\rvert$ when $X$ is used, and $b_i' = \lvert B_i' \rvert$ when $X'$ is used. Observe that the view may be reconstructed from the vector $(\sum_{i=1}^{2m+1} \gamma_i, b_1+\gamma_1, \ldots, b_{2m+1} + \gamma_{2m+1} )$ by a post-processing function that depends on $\boldsymbol{\gamma}^\calA$. This itself is a post-processing of $(-\sum_{i=1}^{2m+1} b_i, b_1 + \gamma_1, \ldots, b_{2m+1}+ \gamma_{2m+1})$, which is again a post-processing of $\boldsymbol{b}+\boldsymbol{\gamma}$ (since the first coordinate may be dropped as it is simply $-n$). The vector $\boldsymbol{b}'$ is equal to $\boldsymbol{b}$ except for in exactly two coordinates $s,t$, where it holds that $b_s' =b_s+1$ and $b_t'=b_t-1$. Thus, $(\epsilon, \delta)$-DP between $\boldsymbol{b} + \boldsymbol{\gamma}$ and $\boldsymbol{b}' + \boldsymbol{\gamma}$ then follows from Lemma~\ref{lem:cd}. 
\end{proof}
\subsection{Proof of Theorem \ref{thm:name}}
\begin{proof}
    The algorithm may be viewed as a composition of running $\FSEM$ on the subsampled dataset with privacy parameters $(\epsilon_1, \delta)$, running $\FBuc$ on $X$ with parameters $(\epsilon_2, \delta)$, and a parallel composition of $\FSEM$ on the partitions $B_1, \ldots, B_{2m}$ of $X$ with parameters $(\epsilon_3, \delta)$. Using the simple composition theorem for approximate DP and the assumption on that $\varepsilon_1,\varepsilon_2,\varepsilon_3 = O(1)$ we get the stated bound.

\end{proof}

\subsection{Proof of Theorem \ref{thm:security}}\label{app:security}
\begin{proof}
We begin by first providing a construction for our
simulator $\textsf{Sim}$ for the ideal world. We work in the hybrid
model, where invocations of sub-protocols can be replaced
with that of the corresponding functionalities, as long as the
sub-protocol is proven to be secure. We will operate in the
$(\F_{\textsf{Sort}},\F_{\textsf{Shuffle}})$-hybrid model, where we assume the existence of a
secure protocol 
which realizes the desired ideal functionalities  $\F_{\textsf{Sort}}$ and $\F_{\textsf{Shuffle}}$.
For the ease of exposition, we start with the case of honest clients.
 WLOG let us assume $\Ser_1$ is the corrupt server. The simulator sends  random values as the shares of the honest clients to the $\Ser_1$. All the messages from $\Ser_0$ to $\Ser_1$ can be simulated using dummy values. If any intermediate values are revealed (e.g., comparison outcomes), the simulator uses access to the ideal functionality  to construct consistent values. In particular, $\Sim$ samples the noise for creating the masking array $L^0$ from the correct distribution, and sends the shares to $\calA$. $\Sim$ receives the corresponding shares of the $\calA$'s masking array.  If the shares are not valid, $\Sim$ sends abort to $\mathcal{F}$.  Next, $\Sim$ receives the total number of dummy records from $\F_\name$ and creates the dummy records accordingly, that passes the check.   It, in turn, receives the shares of $\calA$'s dummy records. In case, the shares are invalid, $\Sim$ sends abort to $\F_{\name}$. Else, $\Sim$ sends a vector $\boldsymbol{\eta^\calA}$ such that $\sum_{i=1}^{2m+1}\eta^\calA_i = t_\calA$ which is the total number of dummy records received from $\calA$. Using $\sum_{i=1}^{2m} \eta_i$  and $cnt$, from $\F_{\name}$, the simulator can distribute the noise it generated over the buckets to be consistent with $cnt$.
\end{proof}

\section{Utility and Runtime Proofs}
In this section, for a dataset $X$ and an element $z \in D$, we let $\rank_X(z) = \sum_{x \in X} \mathbf{1}[x \leq z]$.
\subsection{Proof of Theorem~\ref{thm:slice-utility}}
By the utility guarantee of $\CC$ (Lemma~\ref{lem:co}), the $\Delta_i$'s in $\FSEM$ satisfy $|\Delta_i| \leq \frac{24}{\epsilon} \log(m) \ln(\frac{2m}{\delta})$, and thus they contribute this much to the error. By the utility guarantee of the exponential mechanism~\cite{mcsherry2007mechanism}, all quantile estimates will have error at most $\frac{12}{\epsilon} \ln (\frac{m|D|}{\beta})$. The total error is the sum of the two values.

\subsection{Proof of Theorem~\ref{thm:bucket-util}}
\begin{proof}
    By the utility guarantee of $\CCd$ (Lemma~\ref{lem:cd}), we know, for all $1 \leq j \leq m$, $|2j\tau-\sum_{i=1}^j \gamma^\calA_i|> \tau$, with a similar guarantee holds for $\gamma_i$. Utility then follows from the definition $cnt_i=|B_i|+\gamma_i + \gamma_i^\calA$.
\end{proof}
\subsection{Proof of Theorem~\ref{thm:main-util}}\label{app:utility}
For utility, we make use of the following convergence inequality.
\begin{lemma}\label{lem:dkw}
    (Dvoretzky–Kiefer–Wolfowitz inequality, e.g.~\cite{kosorok2008introduction}): Suppose $\{x_1, \ldots, x_{k}\}$ are i.i.d. samples from a distribution $X$ on a domain $\calX$. Define $F_k(x) = \frac{1}{k} \sum_{i=1}^k \mathbf{1}[x_i \leq x]$ and $F_X(x) = \Pr[X \leq x]$. Then, for all $t \geq 0$, we have
    \[
        \Pr[\sup_{x \in \calX}|F_k(x) - F_X(x)| \geq t] \leq 2e^{-2kt^2}.
    \]
\end{lemma}
This allows us to prove the following result important for utility:
\begin{lemma}\label{lem:initial-buckets}
    For a quantile $q \in Q$, let its associated quantile set be $Q_{j_q}$ in $\F_{\name}$. With probability $1-2\beta$, we have for all quantiles $q \in Q$ that 
    \begin{multline*}
        -\frac{12}{\epsilon_2}\log(m)\log(\tfrac{2m}{\beta}) \leq qn - 8\tau j_q + \sum_{i=1}^{2j_q-1}cnt_i \\ \leq cnt_{2j_q} + \frac{12}{\epsilon_2}\log(m)\log(\tfrac{2m}{\beta}),
    \end{multline*}
    where $\tau = \frac{6}{\epsilon_2} \log (m) \log(\frac{2m}{\delta})$.
\end{lemma}
\begin{proof}
    Let $Y = \Xsam$ and $\rank_X(w) = \sum_{x \in X} \mathbf{1}[w \leq x]$, and $k = n^{2/3}$ (the size of the subsample).
    We will first show that $\rank_X(v_{2j_q-1}') \leq qn \leq \rank_X(v_{2j_q}')$. By Lemma~\ref{lem:dkw}, we know that for any $w \in D$,
    \[
        \left \lvert \frac{1}{k} \rank_Y(w) - \frac{1}{n} \rank_X(w) \right \rvert \leq \sqrt{ \frac{\ln(2/\beta)}{2k} }
    \]
    with probability at least $1-\beta$. Furthermore, Theorem~\ref{thm:slice-utility} guarantees that the quantile estimates $V' = \FSEM(Y, Q')$ have an error guarantee of $c = \frac{12\log({|D|m/\beta})}{\epsilon_1} + \frac{24\log m \log(\frac{2m}{\beta})}{\epsilon_1}$ with probability at least $1-\beta$. We can write rewrite this error guarantee as, for all $q \in Q$,
    \begin{gather*}
        \frac{1}{k}\rank_Y(v_{2j_q-1}') \leq q'_{2j_q-1} + \frac{c_2}{k} \leq q - \alpha + \frac{c}{k} \\
        \frac{1}{k}\rank_Y(v_{2j_q}') \geq q_{2j_q}' - \frac{c_2}{k} \geq q + \alpha - \frac{c}{k},
    \end{gather*}
    where the second inequalities hold because $q'_{2j_q-1} \leq q-\alpha$ and $q'_{2j_q} \geq q+\alpha$ due to the initial quantile merging process. By the triangle inequality, we have
    \begin{gather*}
         \frac{1}{n} \rank_X(v_{2j_q-1}') \leq q - \alpha + \frac{c}{k} + \sqrt{\frac{\ln(2/\beta)}{2k}} \leq q
         \\
        \frac{1}{n}\rank_X(v_{2j_q}') \geq q +\alpha - \frac{c}{k} - \sqrt{ \frac{\ln(2/\beta)}{2k} } \geq q,
    \end{gather*}
    where the final inequality comes from the choice $\alpha = \frac{c}{k} + \sqrt{\frac{\ln(2/\beta)}{2k}}$.
    
    To finish the proof, let $B_i = \{x \in X \vert v_{i-1}' \leq x < v_{i}'\}$ denote the data partition defined in $\F_{\name}$ before dummy users are added. We have the following:
    \begin{align*}
        \sum_{i=1}^{2j_q-1} cnt_i
        &\leq 4\tau(2j_q-1) + \sum_{i=1}^{2j_q-1} |B_i| + \frac{12}{\epsilon_2}\log(m)\log(\tfrac{2m}{\beta}) \\
        &\leq \rank_{X}(v_{2j_q-1}') + 4\tau(2j_q-1) + \frac{12}{\epsilon_2}\log(m)\log(\tfrac{2m}{\beta}) \\
        &\leq qn + 8j_q\tau + \frac{12}{\epsilon_2}\log(m)\log(\tfrac{2m}{\beta})
    \end{align*}
    where the first step follows from Theorem~\ref{thm:bucket-util}, and the second step follows from the fact that the $B_i$ partition $X$ and $v'_{2j_q-1} \in B_{2j_q}$. Similarly,
    \begin{align*}
        \sum_{i=1}^{2j_q} cnt_i &\geq 4\tau(2j_q) + \sum_{i=1}^{2j_q} |B_i| - \frac{12}{\epsilon_2}\log(m)\log(\tfrac{2m}{\beta}) \\
        &\geq \rank_X(v_{2j_q}') + 8\tau j_q-\frac{12}{\epsilon_2}\log(m)\log(\tfrac{2m}{\beta}) \\
        &\geq qn + 8j_q\tau - \frac{12}{\epsilon_2}\log(m)\log(\tfrac{2m}{\beta}).
    \end{align*}
    The result follows by rearranging.
\end{proof}
We are then able to prove Theorem~\ref{thm:main-util}.

\begin{proof}
    Let $v$ be an arbitrary returned quantile, and suppose it corresponds to rescaled quantile $q' \in Q_{j}'$ and quantile $q$ in quantile partition $Q_{j}$. For $1 \leq i \leq 2m$, let $B_{i} = \{x \in X \vert v_{i-1}' \leq x \leq v_{i}'\}$, and $\tilde{B}_i$ denote $B_i$ after dummy users are added. Observe $|\tilde{B}_i| = cnt_i$.

    Because the bins $B_1, \ldots, B_{2m}$ partition $X$, and $v$ is guaranteed to lie in $B_{2j}$, we have
    \begin{align*}
        \rank_{X}(v) &= \sum_{i=1}^{2j-1} |B_i| + \rank_{B_{2j}}(v) \\
        &= \sum_{i=1}^{2j} |B_i| + \rank_{\tilde{B}_{2j}}(v) - cnt_{2j},
    \end{align*}
    where the second step follows because $cnt_{2j} - |B_{2j}|$ dummy elements were added to obtain $\tilde{B}_{2j}$. This implies that
    \begin{align*}
        \err_{X,q}(v) &= \lvert \rank_X(v) - qn \rvert \\
        &= \left \lvert \rank_{\tilde{B}_{2j}}(v) +  \sum_{i=1}^{2j} |B_i| - cnt_{2j} - qn \right \rvert \\ 
        &\leq  \lvert \rank_{\tilde{B}_{2j}}(v) - q' cnt_{2j}\rvert \\ 
        &\qquad + \left \lvert q'cnt_{2j} + \sum_{i=1}^{2j} |B_i| - cnt_{2j} - qn \right \rvert
    \end{align*}
    By Theorem~\ref{thm:slice-utility}, we know that with probability at least $1 - 2\beta$ the first term is at most
    \[
    \mathsf{err}_{\tilde{B}_{2j}, q'}(v) \leq \frac{12\log(\frac{|D|m}{\beta})}{\epsilon_3} + \frac{24\log (m) \log (\frac{2m}{\beta})}{\epsilon_3}.
    \]
    To control the second term, we may plug in the definition of $q'$:
    \[
        q' \approx \frac{qn + 8 \tau j - 2\tau - \sum_{i=1}^{2j-1} cnt_i}{cnt_{2j}},
    \]
    where by Lemma~\ref{lem:initial-buckets}, the approximation holds up to an additive factor of $\frac{12 \log(m) \log(2m/\beta)}{\epsilon_2 cnt_{2j}}$ with probability $1-2\beta$. Plugging in, we see that
   \begin{align*}
   &\left \lvert q'cnt_{2j} + \sum_{i=1}^{2j} |B_i| - cnt_{2j} - qn \right \rvert \\
       &\leq \left\lvert \sum_{i=1}^{2j}|B_i| - \sum_{i=1}^{2j}cnt_i + 8 \tau j \right\rvert + \frac{12}{\epsilon_2} \log m \log(\tfrac{2m}{\beta})\\
       &\leq \frac{24}{\epsilon_2} \log m \log(\tfrac{2m}{\beta}),
   \end{align*}
   where the final line follows from Theorem~\ref{thm:bucket-util}. Thus, we have shown the desired error guarantee, and a success probability of $1-6\beta$.
\end{proof}
\subsection{Proof of Lemma \ref{lemma}}\label{app:lemma}

\begin{proof}
 $\eta_i^1$ elements from the very end of $\bar{S}_i$ are shifted to the very beginning after the sort. This is because since our domain is $[0,d]$, all these values are less than the inputs. Similarly,  $\eta_i^0$ elements from the beginning of $\bar{S}_i$ are shifted to the very end after the sort. As a result,  $\hat{S}_i$ can be written as 
    \[
        A + \Xs[a_i-w+\eta^0_i: b_i +w- \eta^1_i] + B, 
    \]
    where $+$ denotes concatenation, $A$ is a dataset of length $\eta^{1}_i$, and $B$ is a dataset of length $\eta^0_{i}$. Hence, by removing $w$ elements from both the left and right of the above list, we obtain 
    \begin{align*}
    \Xs[a_i+\eta^0_i-\eta^1_i, b_i \eta^0_i-\eta^1_i]
     \end{align*}
     which concludes the proof.
    
\end{proof}

\subsection{Proof of Theorem \ref{thm:secure comp}}
\begin{proof}
    The comparisons are introduced by the initial call to $\FSEM$, the call to $\FBuc$, and finally each of up to $m$ calls to $\FSEM$ on the quantile partitions $Q_i$. By Theorem~\ref{coro:partial-sort}, the number of comparisons used by the the first call to $\FSEM$ is $O(k \log (m) + m(w+h) \log n)$. For nearly all parameter regimes, the lefthand term will dominate. Let $m_i = |\tilde{Q}_i|$; we have that $|B_i| \leq 2nm_i\alpha$. By the same theorem, the running time of the later calls to $\FSEM$ is $O(|B_i| \log (m_i) + m_i(w+h)\log(|B_i|))$, which will again be dominated by $O(|B_i| \log (m_i))$. Thus, the total running time will be at most
    \begin{align*}
        & k\log(m) + \sum_{i=1}^m {2nm_i}\log(m) \alpha \\
        &\leq k \log(m) + {2nm} \log(m) \alpha \\
        &= k \log(m) + 2nm \Bigg(\sqrt{\frac{\ln(1/\beta)}{k}} \\ &\qquad+ \frac{\ln(Dm/\beta)}{k\epsilon_1} + \frac{ \log(m) \log(m/\beta)}{k \epsilon_1}\Bigg)\log(m).
    \end{align*}
    This expression is minimized by taking 
    \begin{multline*}
    k = (nm)^{2/3}  \max\Bigg\{\ln (\tfrac{1}{\beta})^{1/3}, \\ \sqrt{\frac{\ln(Dm/\beta) + \log(m)\log(m/\beta)}{\epsilon_1 (nm)^{1/3}}} \Bigg\};
    \end{multline*}
    the total number of comparisons is then bounded by $O(k)$.
    Finally, it is easy to see that the call to $\FBuc$ requires at most $n \log (2m)$ due to the binary search.
\end{proof}
\section{Details from Partial Sorting}\label{app:partial sort}
In this section, we prove general self-contained result about partial sorting, which implies Theorem \ref{coro:partial-sort}. This result was provided to us by~\cite{Aanders}.
\begin{theorem}\label{thm:partial-sort} Let $A\subset [n]$ be a disjoint union of intervals $A=I_1\cup\cdots\cup I_k$  such that no two of these can be combined into a larger interval. 
Denote by $\alpha=|A|$ and $\beta=|[n]\setminus A|$. The expected number of comparisons required to sort  a dataset of length $n$, allowing elements with ranks in $A$
to remain unordered, is at  most $\alpha(5+\ln(nk/\alpha))+4\beta \ln n$. 
\label{thm:partialsort}\end{theorem}

Consider the following comparison-based sorting problem. Let $X=\{x_1,\dots, x_n\}$ be a set of size $n$ and $\prec$ a linear order on $X$. Let $\pi:[n]\to[n]$ be the sorted order of $X$, i.e., $x_{\pi(1)}\prec \cdots\prec x_{\pi(n)}$. Let $A\subset [n]$ be some set and write $A=I_1\cup\cdots\cup I_k$ as a disjoint union of intervals such that no two of these can be combined into a larger interval.
Our goal is to find a \emph{partial sorting} $\rho:[n]\to [n]$  satisfying that $\rho(i)=\pi(i)$ for all $i\in [n]\setminus A$, and for each $t=1,\dots, k$, $\{\pi(i):i\in I_t\}=\{\rho(i):i\in I_t\}$. In other words, restricted to any interval $I_t$, $\rho$ is only required to get the elements right, not their actual sorted order. We assume access to comparisons of the form "is $x_i\prec x_j$?", and are interested in the number of comparisons needed to find such a $\rho$. Note that $A=\emptyset$ corresponds to standard comparison-based sorting. We would like to show that when $A$ is large, we can get away with using significantly fewer comparisons. 

Our algorithm is similar to Quicksort but with a slight modification. Let $(X,A,\prec)$ be some problem instance, with $|X|=n$. If $|X|=|A|$, we perform no more comparisons and let $\rho: [n]\to [n]$ be arbitrary. Otherwise, we pick a random pivot $x\in X$, form the two subsets $X_1=\{y\in X:y\prec x\}$ and $X_2=\{y\in X: x\prec y\}$, define $n_1=|X_1|$, $n_2=|X_2|$, $A_1=A\cap [n_1]$, and $A_2=A\cap ([n]\setminus [n_1+1])$, and recurse on the two subinstances $(X_1,A_1,\prec)$ combining the solutions similar to Quicksort.

\begin{theorem}
Denote by $\alpha=|I|$ and $\beta=|[n]\setminus I|$. The expected number of comparisons performed by the modified Quicksort algorithm is at most $\alpha(5+\ln(nk/\alpha))+4\beta \ln n$.
\end{theorem}
\begin{proof}
Our analysis is similar to that of Quicksort in that we for each pair $x,y\in X$ calculate the probability that $x$ is ever compared to $y$ and sum over all pairs $x,y$. The main difference it that for standard Quicksort, the probability that $x_{\pi(i)}$ and $x_{\pi(j)}$ are ever compared is $\frac{2}{j-i+1}$ but in our setting we get that in some cases the probability is even smaller. 

Denote $\ell_t=|I_t|$ and write $I_t=\{a_t,a_t+1,\dots, a_t+\ell_t-1\}$ for $t=1,\dots,k$. For $i<j$, we denote the $B_{i,j}$ the event that $x_{\pi(i)}$ and $x_{\pi(j)}$ are ever compared. In general $\Pr[B_{i,j}]\leq \frac{2}{j-i+1}$ by a similar logic to Quicksort: $x_{\pi(i)}$ and $x_{\pi(j)}$ will only be compared if we ever select a pivot in the set $\{x_{\pi(i)},\dots, x_{\pi(j)}\}$ and that pivot is either $x_{\pi(i)}$ or $x_{\pi(j)}$. However, as we will show, for some values of $i$ and $j$ it is not even certain that we will ever select such a pivot (if at some point they are both contained in a subinstance where we don't recurse but instead pick an arbitrary $\rho$). In particular, suppose that $i<j$ are such that $\pi(i)$ and $\pi(j)$ both lie in the \emph{same} set $I_t$. We claim that
\begin{align*}
\Pr[B_{i,j}]&=\frac{i-a_t}{\ell_t}\cdot\frac{2}{a_t+\ell_t-i} +\\
&+\frac{a_t+\ell_t-j-1}{\ell_t}\cdot\frac{2}{j-a_t+1}+\frac{2}{\ell_t}.
\end{align*}
Here the first term accounts for the probability that the first randomly selected pivot in $I_t$ is smaller $i$ and that the second time we select a pivot in $\{i,i+1,\dots, a_t+\ell_t-1\}$ it is either $i$ or $j$. The second term symmetrically accounts for the probability that the first randomly selected pivot in $I_t$ is larger than $j$ and that the second time we select a pivot in $\{a_t,a_t+1,\dots, j\}$ it is either $i$ or $j$. Finally the third term accounts for the probability that the first pivot selected in $I_t$ is either $i$ or $j$. These are the only options for $x_{\pi(i)}$ and $x_{\pi(j)}$ to ever be compared. To see this, note that if the first sampled pivot in $I_t$ is below $i$, and the first sampled pivot between $i$ and $a_t+\ell-1$ is not $i$ or $j$, then $x_{\pi(i)}$ and $x_{\pi(j)}$ are never compared: Indeed, if (1) the pivot is between $i$ and $j$, they are placed in different subinstances, and if (2) the pivot is above $j$, then $A_1=[n_1]$ in the left subinstance, so we don't recurse. The case where the first pivot is above $j$ is similar.

Summing this quantity over $i,j$ with $a_t\leq i <j<a_t+\ell_t$ gives the expected number of comparisons between elements within $\{x_{\pi(i)}:i\in I_t\}$. Note that 
\begin{align*}
\sum_{a_t\leq i <j<a_t+\ell_t} &\frac{i-a_t}{\ell_t}\cdot\frac{2}{a_t+\ell_t-i}=\sum_{0\leq i <j<\ell_t} \frac{i}{\ell_t}\cdot\frac{2}{\ell_t-i}\\
&=\sum_{0\leq i\leq \ell_t-2}(\ell_t-i-1)\frac{i}{\ell_t}\cdot\frac{2}{\ell_t-i}\\
&\leq \sum_{0\leq i\leq \ell_t-2}\frac{2i}{\ell_t}\leq \ell_t
\end{align*}
and similarly,
\[
\sum_{a_t\leq i <j<a_t+\ell_t} \frac{a_t+\ell_t-j-1}{\ell_t}\cdot\frac{2}{j-a_t+1}\leq \ell_t
\]
We can thus bound the expected number of comparisons between pairs of elements in $\{x_\pi(i): i \in I_t\}$ by
\begin{align}\label{eq:within}
\sum_{a_t\leq i <j<a_t+\ell_t}\Pr[B_{i,j}]\leq 3\ell_t.
\end{align}
This accounts for comparisons within $\{x_{\pi(i)}:i\in I_t\}$. 
For pairs $i,j$ where $\pi(i)\in I_t$ and $\pi(j)\notin I_t$, we use that $\Pr[B_{i,j}]=\frac{2}{|j-i|+1}$ (in fact, we only need that this is an upper bound on $\Pr[B_{i,j}]$ but it is easy to check that equality holds). Thus, we can bound the expected number of comparisons between all such $i,j$ by
\begin{align}\label{eq:cross}
\sum_{\substack{i:\pi(i)\in I_t \\ j: \pi(j) \in  [n]\setminus I_t}}&\Pr[B_{i,j}]=\\
&=\sum_{i=a_t}^{\ell_t+a_t-1}\left(\sum_{j=1}^{a_{t}-1}\frac{2}{i-j+1}+\sum_{j=a_t+\ell_t}^{n}\frac{2}{j-i+1}\right)
\end{align}
Denoting by $H_m=\sum_{i=1}^m\frac{1}{i}$, and recalling the formula $\sum_{i=1}^m H_i=(m+1)H_m-m$ we thus obtain that 
\begin{align*}
\sum_{i=a_t}^{\ell_t+a_t-1}\left(\sum_{j=1}^{a_{t}-1}\frac{2}{i-j+1}\right)&=\sum_{i=a_t}^{\ell_t+a_t-1}H_i-H_{i-a_t+1}\\
&\leq \ell_t H_n-\sum_{i=1}^{\ell_t}H_i\\
&=\ell_t(H_n-H_{\ell_t})+\ell_t-H_{\ell_t}.
\end{align*}
And the same bounds symmetrically holds for the second term in~\ref{eq:cross}. Summing~\ref{eq:within} and~\ref{eq:cross} over all $t$, we get that the expected number of comparisons involving \emph{any} element $x_{\pi(i)}$ where $\pi(i)$ lies in some $I_t$ can be upper bounded by
\begin{align}\label{eq:in-and-out}
4\alpha+\sum_{t=1}^k \ell_t(H_n-H_{\ell_t}).
\end{align}
Now it is always the case that $H_m\geq \ln(m)$, and so
\[
\sum_{t=1}^k\ell_t H_{\ell_t}\geq \sum_{t=1}^k \ell_t\ln(\ell_t)\geq \alpha \ln(\alpha/k)
\]
and thus, we can upper bound the expression in~\ref{eq:in-and-out} by
\[
\alpha(4+H_n-\ln(\alpha/k))\leq\alpha(5+\ln(nk/\alpha)).
\]
This accounts for the number of comparisons involving elements $x_{\pi(i)}$ where $\pi(i)\in I$. Finally, if $\pi(i)\in [n]\setminus I$, it follows from the standard Quicksort analysis that the expected number of comparisons involving $x_{\pi(i)}$ is at most $4\ln n $, so the expected number of such comparisons is at most $4\beta \ln n$. This completes the proof of the theorem.
\end{proof}

\end{document}